\documentclass[a4paper,authorcolumns,USenglish,cleveref,autoref,thm-restate]{lipics-v2021}
\usepackage{xspace}
\usepackage[dvipsnames]{xcolor}

\definecolor{blue}{RGB}{0,50,200}
\definecolor{magenta}{RGB}{255,0,255}
\hypersetup{colorlinks,linkcolor=blue,urlcolor=blue,citecolor=magenta}

\nolinenumbers

\hideLIPIcs
\title{A 3.3904-Competitive Online Algorithm for List Update with Uniform Costs}

\author{Mateusz Basiak}{University of Wrocław}{}{https://orcid.org/0009-0009-0210-6451}{}
\author{Marcin Bienkowski}{University of Wrocław}{}{https://orcid.org/0000-0002-2453-7772}{}
\author{Martin Böhm}{University of Wrocław}{}{https://orcid.org/0000-0003-4796-7422}{}
\author{Marek Chrobak}{University of California, Riverside}{}{https://orcid.org/0000-0002-8673-2709}{}
\author{Łukasz Jeż}{University of Wrocław}{}{https://orcid.org/0000-0002-7375-0641}{}
\author{Jiří Sgall}{Charles University}{}{https://orcid.org/0000-0003-3658-4848}{}
\author{Agnieszka Tatarczuk}{University of Wrocław}{}{https://orcid.org/0009-0008-3849-9165}{}

\authorrunning{M. Basiak, M. Bienkowski, M. Böhm, M. Chrobak, Ł. Jeż, J. Sgall and A. Tatarczuk} 
\Copyright{Mateusz Basiak, Marcin Bienkowski, Martin Böhm, Marek Chrobak, Łukasz Jeż, Jiří Sgall, Agnieszka Tatarczuk}
\funding{Suppported by Polish National Science Centre grants 2022/45/B/ST6/00559, \\
2020/39/B/ST6/01679 and 2022/47/D/ST6/02864,
National Science Foundation grant CCF-2153723, and by grant 24-10306S of GA~\v{C}R.\\
\phantom{x}\\
This work is licensed under the Creative Commons Attribution (CC BY) License.}

\EventEditors{Anne Benoit, Haim Kaplan, Sebastian Wild, and Grzegorz Herman}
\EventNoEds{4}
\EventLongTitle{33rd Annual European Symposium on Algorithms (ESA 2025)}
\EventShortTitle{ESA 2025}
\EventAcronym{ESA}
\EventYear{2025}
\EventDate{September 15--17, 2025}
\EventLocation{Warsaw, Poland}
\EventLogo{}
\SeriesVolume{351}
\ArticleNo{74}

\newtheorem{assumption}[theorem]{Assumption}
\newtheorem{compresult}[theorem]{Computational Result}
\crefname{assumption}{Assumption}{Assumptions}

\crefname{fact}{Fact}{Facts}

\newcommand{\etal}{\emph{et al.\,}}

\newcommand{\probLUPstandard}{\ensuremath{\textsf{LUP}_\mathrm{S}}\xspace}
\newcommand{\probLUPuniform}{\ensuremath{\textsf{LUP}_1}\xspace}
\newcommand{\OPT}{\textsl{\textsc{Opt}}\xspace}
\newcommand{\ALG}{\textsl{\textsc{Fpm}}\xspace}
\newcommand{\ALGLONG}{\textsl{\textsc{Full-Or-Partial-Move}}\xspace}
\newcommand{\MTF}{\textsl{\textsc{Mtf}}\xspace}
\newcommand{\MTFLONG}{\textsl{\textsc{Move-To-Front}}\xspace}
\newcommand{\AALG}{\textsl{\textsc{Alg}}\xspace}

\newcommand{\BIT}{\textsl{\textsc{Bit}}\xspace}
\newcommand{\DBIT}{\textsl{\textsc{Dbit}}\xspace}
\newcommand{\FRONT}{\textsl{\textsc{Front}}\xspace}

\newcommand{\COUNTER}{\textsl{\textsc{Counter}}\xspace}
\newcommand{\RR}{\textsl{\textsc{RandomReset}}\xspace}
\newcommand{\TimeStamp}{\textsl{\textsc{TimeStamp}}\xspace}

\renewcommand{\P}{\mathcal{P}}
\newcommand{\Q}{\mathcal{Q}}
\newcommand{\set}[1]{\{#1\}}
\newcommand{\G}[1]{\ensuremath{G_\mathrm{#1}}\xspace}

\newcommand{\half}{\textstyle{\frac{1}{2}}}
\newcommand{\oneeighth}{\textstyle{\frac{1}{8}}}
\newcommand{\onesixteenth}{\textstyle{\frac{1}{16}}}

\newcommand{\thirteenfourths}{\textstyle{\frac{13}{4}}}

\newcommand{\braced}[1]{{ \left\{ {#1} \right\} }}

\begin{document}

\maketitle


\begin{abstract}
We consider the List Update problem where the cost of each swap is assumed to
be~1. This is in contrast to the ``standard'' model, in which an algorithm is
allowed to swap the requested item with previous items for free. We construct an online
algorithm \textsc{Full-Or-Partial-Move} (\ALG), whose competitive ratio is at most 
$3.3904$, improving over the previous best known bound of $4$.	
\end{abstract}


\section{Introduction}
\label{sec:introduction}

\subparagraph*{The List Update problem.}
In the online \emph{List Update
problem}~\cite{2008_albers_lup_EfoA,2008_kamali_lup_EofA,2013_kamali_lopez-ortiz_lup_survey},
the objective is to maintain a set of items stored in a linear list in response
to a sequence of access requests. The cost of accessing a requested item is
equal to its distance from the front of the list. After each request, an
algorithm is allowed to rearrange the list by performing an arbitrary number of
swaps of adjacent items. In the model introduced by Sleator and Tarjan in their
seminal 1985 paper on competitive
analysis~\cite{1985_sleator_tarjan_lup_paging}, an algorithm can repeatedly swap
the requested item with its preceding item at no cost. These swaps are called
\emph{free}. All other swaps are called \emph{paid} and  have cost $1$ each. 
As in other problems involving self-organizing data
structures~\cite{2005_albers_westbrook_self_organizing}, 
the goal is to construct an \emph{online} algorithm, i.e., 
operating without the knowledge of future requests.
The cost of such an algorithm is compared to the cost of 
the optimal \emph{offline} algorithm; the ratio of the two costs is called
the \emph{competitive ratio} and is subject to minimization.

Sleator and Tarjan proved that the algorithm \MTFLONG ($\MTF$), which after each
request moves the requested item to the front of the list, is
$2$-competitive~\cite{1985_sleator_tarjan_lup_paging}. This ratio is known to be
optimal if the number of items is unbounded. Their work was the culmination of
previous extensive studies of list updating, including experimental results,
probabilistic approaches, and earlier attempts at amortized analysis
(see~\cite{1985_bentley_mcgeoch_self_organizing_heuristics} and the references
therein).

As shown in subsequent work, $\MTF$ is not unique --- there are other strategies
that achieve ratio $2$, such as \TimeStamp~\cite{1998_albers_improved_lup} or
algorithms based on work functions~\cite{1999_anderson_lup_work_function}. In fact, 
there are infinitely many algorithms that achieve ratio
$2$~\cite{2002_bachrach_etal_theory_and_practice_lup,2013_kamali_lopez-ortiz_lup_survey}.


\subparagraph*{The uniform cost model.}
Following~\cite{1994_reingold_etal_randomized_lup}, we will refer to the cost
model of~\cite{1985_sleator_tarjan_lup_paging} as \emph{standard}, and we will
denote it here by $\probLUPstandard$. This model has been questioned in the
literature for not accurately reflecting true costs in some
implementations~\cite{1994_reingold_etal_randomized_lup,2000_munro_competitiveness_linear_search,2000_martinez_roura_competitiveness_mtf,2012_golynski_etal_lup_mrm_cost_model},
with the concept of free swaps being one of the main concerns.

A natural approach to address this concern, considered in some later studies
(see,
e.g.,~\cite{1994_reingold_etal_randomized_lup,2005_albers_westbrook_self_organizing,2013_kamali_lopez-ortiz_lup_survey,2020_albers_janke_new_bounds_randomized,2024_azar_etal_lup_delays_time_windows}),
is simply to charge cost $1$ for \emph{any} swap. We will call it here the
\emph{uniform cost model} and denote it $\probLUPuniform$. A general lower bound
of~$3$ on the competitive ratio of deterministic algorithms for
$\probLUPuniform$ was given by
Reingold~\etal~\cite{1994_reingold_etal_randomized_lup}. Changing the cost of
free swaps from $0$ to $1$ at most doubles the cost of any algorithm, so $\MTF$
is no worse than $4$-competitive for $\probLUPuniform$. Surprisingly, no
algorithm is known to beat \MTF, i.e., achieve ratio lower than
$4$.\footnote{The authors of~\cite{2013_kamali_lopez-ortiz_lup_survey} claimed
an upper bound of $3$ for $\probLUPuniform$, but later discovered that their
proof was not correct (personal communication with S.~Kamali).}.


\subparagraph*{Our main result.}
To address this open problem, we develop an online algorithm \ALGLONG (\ALG) for
$\probLUPuniform$ with competitive ratio $\oneeighth \!\cdot\! (23+\sqrt{17})
\approx 3.3904$, significantly improving the previous upper bound of $4$. 

Our algorithm $\ALG$ remains $3.3904$-competitive even for the \emph{partial
cost} function, where the cost of accessing location $\ell$ is $\ell-1$, instead
of the \emph{full cost} of~$\ell$ used in the original definition of List
Update~\cite{1985_sleator_tarjan_lup_paging}. Both functions have been used in
the literature, depending on context and convenience. For any online algorithm,
its partial-cost competitive ratio is at least as large as its full-cost ratio,
although the difference typically vanishes when the list size is unbounded. We
present our analysis in terms of partial costs.

\ALG remains $3.3904$-competitive also in the 
dynamic scenario of List Update that allows operations of insertions and
deletions, as in the original definition in~\cite{1985_sleator_tarjan_lup_paging}
(see~\autoref{sec:dynamic}).


\subparagraph*{Technical challenges and new ideas.}
The question whether ratio $3$ can be achieved remains open. We also do not know if
there is a \emph{simple} algorithm with ratio below $4$.
We have considered some natural adaptions of \MTF and other algorithms that are $2$-competitive for $\probLUPstandard$,
but for all we were able to show lower bounds higher than $3$ for $\probLUPuniform$.

As earlier mentioned, $\MTF$ is $4$-competitive for $\probLUPuniform$. It is also easy to
show that its ratio is not better than $4$: repeatedly request the last item
in the list. Ignoring additive constants, the algorithm pays twice the length of
the list at each step, while any algorithm that just keeps the list in a fixed order, 
pays only half the length on the average.

The intuition why $\MTF$ performs poorly is that it moves the requested items to front 
``too quickly''. For the aforementioned adversarial strategy against \MTF,
ratio $3$ can be obtained by moving the items to front only
\emph{every other time} they are requested. This algorithm, called \DBIT, is a
deterministic variant of algorithm \BIT from~\cite{1994_reingold_etal_randomized_lup} and 
a~special case of algorithm
\COUNTER in~\cite{1994_reingold_etal_randomized_lup}, and it
has been also considered in~\cite{2013_kamali_lopez-ortiz_lup_survey}. 
In~\autoref{sec:lower}, we show that \DBIT is not
better than $4$-competitive in the partial cost model, 
and not better than $3.302$-competitive in the full cost model. 

One can generalize these approaches by considering a more general
class of algorithms that either leave the requested item at its current location or move it to the 
front. In~\autoref{sec:lower} we show that such a strategy cannot achieve ratio better than $3.25$, 
even for just three items. 
A naive fix would be, for example, to always move the requested item half-way towards
front. This algorithm is even worse: its competitive ratio is at least $6$ (see~\autoref{sec:lower}).

We also show (see~\autoref{sec:discussion}) via a computer-aided argument that, for $\probLUPuniform$,
the work function algorithm's 
competitive ratio is larger than $3$, even for lists of length $5$. This is in contrast to the
its performance for the standard model, where it achieves optimal ratio $2$~\cite{1999_anderson_lup_work_function}.

Our algorithm $\ALG$ overcomes the difficulties mentioned above by combining a few new ideas.
The first idea is a more sophisticated choice of the target location for the requested item. 
That is, aside from \emph{full moves} that move the requested items to the list front,
\ALG sometimes performs a \emph{partial move} to a suitably chosen target
location in the list. This location roughly corresponds to the front of the list when this item 
was requested earlier.

The second idea is to keep track, for each pair of
items, of the work function for the two-item subsequence consisting of these
items. These work functions are used in two ways. First, they roughly indicate
which relative order between the items in each pair is ``more likely'' in
an optimal solution. The algorithm uses this information to decide whether 
to perform a full move or a partial move. 
Second, the simple sum of all these pair-based
work functions is a lower bound on the optimal cost, 
which is useful in analyzing the competitive ratio of $\ALG$. 


\subparagraph*{Related work.}
Better bounds are known for randomized algorithms both in the standard model
($\probLUPstandard$) and the uniform cost model ($\probLUPuniform$). For
$\probLUPstandard$, a long line of research culminated in a~$1.6$-competitive
algorithm~\cite{1991_irani_two_results_lup,1994_reingold_etal_randomized_lup,1998_albers_improved_lup,1995_albers_etal_lup_combined_bit_timestamp},
and  a $1.5$-lower bound~\cite{1993_teia_lower_bound_randomized_lup}. The upper
bound of~$1.6$ is tight in the class of so-called projective algorithms, whose
computation is uniquely determined by their behavior on two-item
instances~\cite{2013_ambuhl_lower_bounds_projective_lup}. For $\probLUPuniform$,
the ratio is known to be between
$1.5$~\cite{2020_albers_janke_new_bounds_randomized} and
$2.64$~\cite{1994_reingold_etal_randomized_lup}.

It is possible to generalize the uniform cost function by distinguishing between
the cost of $1$ for following a link during search and the cost of $d\ge 1$ for
a~swap~\cite{1994_reingold_etal_randomized_lup,2020_albers_janke_new_bounds_randomized}
This model is sometimes called the $P^d$ model; in this terminology, our
$\probLUPuniform$ corresponds to $P^1$. While $\MTF$ is $4$-competitive for
$P^1$, it does not generalize in an obvious way to $d>1$. Other known
algorithms for the $P^d$ model (randomized and deterministic \COUNTER, \RR
and \TimeStamp) have bounds on competitive ratios that monotonically decrease with
growing~$d$~\cite{1994_reingold_etal_randomized_lup,2020_albers_janke_new_bounds_randomized}.
In particular, deterministic \COUNTER achieves ratio $4.56$ when $d$ tends to
infinity~\cite{2005_albers_westbrook_self_organizing} and for the same setting
($d \to \infty$) a recent result by Albers and
Janke~\cite{2020_albers_janke_new_bounds_randomized} shows a randomized
algorithm \TimeStamp that is $2.24$-competitive.

A variety of List Update variants have been investigated in the
literature over the last forty years, including models with
lookahead~\cite{1998_albers_lup_lookahead}, locality of
reference~\cite{2008_angelopoulos_etal_lup_locality,
2016_albers_lauer_lup_locality}, parameterized
approach~\cite{2015_dorrigiv_etal_lup_parameterized}, algorithms with
advice~\cite{2016_boyar_etal_lup_advice},
prediction~\cite{2025_azar_predictions}, or alternative cost
models~\cite{2012_golynski_etal_lup_mrm_cost_model,2000_martinez_roura_competitiveness_mtf,2000_munro_competitiveness_linear_search}.
A~particularly interesting model was proposed recently by
Azar~\etal~\cite{2024_azar_etal_lup_delays_time_windows}, where an~online
algorithm is allowed to postpone serving some requests, but is either required
to serve them by a specified deadline or pay a delay penalty.  

In summary, List Update is one of canonical problems in the area of competitive analysis, used
to experiment with refined models of competitive analysis or to study the effects of
additional features. This underscores the need to fully resolve the remaining open questions 
regarding its basic variants, including the question whether ratio $3$ is attainable for $\probLUPuniform$.


\section{Preliminaries}\label{sec:preliminaries}


\subparagraph*{Model.}
An algorithm has to maintain a list of items,
while a sequence $\sigma$ of access requests is presented in an online manner.
In each step $t\ge 1$, the algorithm is presented an access request $\sigma^t$ to an item in the list.
If this item is in a location $\ell$, the algorithm incurs cost $\ell-1$ to access it.
(The locations in the list are indexed $1,2,...$.)
Afterwards,
the algorithm may change the list configuration by performing an arbitrary number 
of swaps of neighboring items, each of cost $1$.

For any algorithm $A$, we denote its cost for processing a sequence $\sigma$ by $A(\sigma)$.
The optimal algorithm is denoted by $\OPT$.


\subparagraph*{Notation.}
Let $\P$ be the set of all unordered pairs of items. For a pair $\set{x,y}
\in \P$, we use the notation $x \prec y$ ($x \succ y$) to denote that $x$ is
before (after) $y$ in the list of an online algorithm. 
(The relative order of $x,y$ may change over time, but
it will be always clear from context what step of the computation we are referring to.)
We use $x \preceq y$ ($x \succeq y$) 
to denote that $x \prec y$ ($x \succ y$) or $x = y$.

For an input $\sigma$ and a pair $\set{x,y} \in \P$, $\sigma_{xy}$ is the
subsequence of $\sigma$ restricted to requests to items $x$ and $y$ only.
Whenever we say that an algorithm serves input $\sigma_{xy}$, we mean that
it has to maintain a list of two items, $x$ and $y$. 


\subsection{Work Functions}

\subparagraph*{Work functions on item pairs.}
For each prefix $\sigma$ of the input sequence, an online algorithm may compute
a so-called \emph{work function $W^{xy}$}, where $W^{xy}(xy)$ (or $W^{xy}(yx)$) 
is the optimal cost of the solution that serves $\sigma_{xy}$ and ends with the list in
configuration~$xy$ (or $yx$). 
(Function $W^{xy}$ also has prefix $\sigma$ as an argument. 
Its value will be always uniquely determined from context.)
The values of $W$ for each step can be computed iteratively 
using straightforward dynamic programming. 
Note that the values of $W$ are non-negative integers and $|W^{xy}(xy) - W^{xy}(yx)| \leq 1$.


\subparagraph*{Modes.}
For a pair $\set{x,y} \in \P$, we define its \emph{mode}
depending on the value of the work function $W^{xy}$ in the current step and the 
mutual relation of $x$ and $y$ in the list of an online algorithm. 
In the following definition we assume that~$y \prec x$.
\begin{itemize}
    \item Pair $\set{x,y}$ is in mode $\alpha$ if $W^{xy}(yx) + 1 = W^{xy}(xy)$. 
    \item Pair $\set{x,y}$ is in mode $\beta$ if $W^{xy}(yx) = W^{xy}(xy)$.
    \item Pair $\set{x,y}$ is in mode $\gamma$ if $W^{xy}(yx) - 1 = W^{xy}(xy)$.
\end{itemize}

For an illustration of work function evolution and associated modes, see
\autoref{fig:work_function}.

\begin{figure}[t]
\centering
\includegraphics[width=0.8\textwidth]{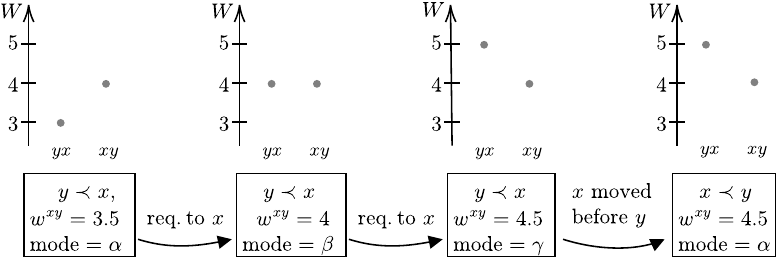}
\caption{The evolution of the work function $W^{xy}$. 
Initially, $W^{xy}$ has its minimum in the state $yx$ and 
$y \prec x$ (in the list of an online algorithm). Thus, the mode of the pair~$\set{x,y}$ is $\alpha$.
Next, because of the two requests to $x$, the value of $W^{xy}(yx)$ is incremented,
while the value at $xy$ remains intact. The mode is thus changed from $\alpha$ to $\beta$
and then to $\gamma$. Finally, when an algorithm moves item $x$ before $y$, 
the mode of the pair $\set{x,y}$ changes to $\alpha$. The value of $w^{xy}$ 
(the average of $W^{xy}(xy)$ and $W^{xy}(yx)$) increases by $\half$ whenever the mode changes due to a request.}
\label{fig:work_function}
\end{figure}

If a pair $\set{x,y}$ is in mode $\alpha$, then the minimum of the work
function $W^{xy}$ is at configuration~$yx$, i.e., the one that has $y$
before~$x$. That is, an online algorithm keeps these items in a way that
``agrees'' with the work function. Note that $\alpha$ is the initial mode
of all pairs.
Conversely, if a pair $\set{x,y}$ is in state $\gamma$, then the minimum of the
work function $W^{xy}$ is at configuration~$xy$. In this case, an online
algorithm keeps these items in a way that ``disagrees'' with the work
function. 


\subsection{Lower Bound on OPT}
\label{sec:list_factoring}

Now we show how to use the changes in the work functions of item
pairs to provide a useful lower bound on the cost of an optimal algorithm. The
following lemma is a~standard and straightforward result of the list partitioning
technique~\cite{1999_anderson_lup_work_function}.\footnote{There are known input
sequences on which the relation of \autoref{lem:pairwise_opt_vs_opt} is not
tight~\cite{1999_anderson_lup_work_function}.}

\begin{lemma}
\label{lem:pairwise_opt_vs_opt}
For every input sequence $\sigma$, it holds that 
$\sum_{\set{x,y} \in \P} \OPT(\sigma_{xy}) \leq \OPT(\sigma)$.
\end{lemma}


\subparagraph*{Averaging work functions.}
We define the function $w^{xy}$ as the average value of the work function $W^{xy}$,
i.e.,
\begin{equation*}
w^{xy} \triangleq \half \cdot (\,  W^{xy}(xy) + W^{xy}(yx)\, ).
\end{equation*}
We use $w_t^{xy}$ to denote the value of $w^{xy}$ after serving the first $t$
requests of $\sigma$, and define $\Delta_t w^{xy} \triangleq w_t^{xy} -
w_{t-1}^{xy}$. The growth of $w^{xy}$ can be related to $\OPT(\sigma)$ in
the following way.

\begin{lemma}
\label{lem:wf_vs_opt}
For a sequence $\sigma$ consisting of $T$ requests, it holds that
$
    \sum_{t=1}^T \sum_{\set{x,y} \in \P} \Delta_t w^{xy} \leq \OPT(\sigma).
$
\end{lemma}

\begin{proof}
We first fix a pair $\set{x,y} \in \P$.
We have $w_0^{xy} = \half$ and $w_T^{xy} \leq \OPT(\sigma_{xy}) + \half$.
Hence, 
$
    \sum_{t=1}^T \Delta_t w^{xy} = w_T^{xy} - w_0^{xy} \leq \OPT(\sigma_{xy}).
$
The proof follows by summing over all pairs $\set{x,y} \in \P$ and invoking 
\autoref{lem:pairwise_opt_vs_opt}.
\end{proof}


\subparagraph*{Pair-based OPT.}
\autoref{lem:wf_vs_opt} gives us a convenient tool to lower bound $\OPT(\sigma)$. We
define the cost of \emph{pair-based \OPT} in step $t$ as $\sum_{\set{x,y} \in \P}
\Delta_t w^{xy}$. For a given request sequence $\sigma$,
the sum of these costs over all steps 
is a lower bound on the actual value of $\OPT(\sigma)$.

On the other hand, we can express $\Delta_t w^{xy}$ (and thus also the pair-based
\OPT) in terms of the changes of the modes of item pairs. See \autoref{fig:work_function}
for an illustration.

\begin{observation}
\label{obs:type_change_causes_cost}
    If a pair $\set{x,y}$ changes its mode due to the request in step $t$ then 
    $\Delta_t w^{xy} = \half$, otherwise $\Delta_t w^{xy} = 0$.
\end{observation}


\section{Algorithm Full-Or-Partial-Move}
\label{sec:algorithm}

For each item $x$, \ALG keeps track of an item denoted $\theta_x$ and called the \emph{target} 
of $x$. At the beginning, \ALG sets $\theta_x = x$ for all $x$. 
Furthermore, at each time, \ALG ensures that $\theta_x \preceq x$.


\smallskip

The rest of this section describes the overall strategy of algorithm \ALG. Our description is top-down, and proceeds in three steps:
\begin{itemize}
\item First we describe, in broad terms, what actions are involved in serving a request,
	including the choice of a move and the principle behind updating target items.
\item Next, we define the concept of states associated with item pairs and their potentials.
\item Finally, we explain how algorithm \ALG uses these potential values to decide how to adjust the
	list after serving the request.
\end{itemize}

This description will fully specify how \ALG works, providing that the potential function on the states is given.
Thus, for any choice of the potential function the competitive ratio of \ALG is well-defined.
What remains is to choose these potential values to optimize the competitive ratio.
This is accomplished by the analysis in~\autoref{sec:analysis} that follows.


\subparagraph*{Serving a request.}
Whenever an item $z^*$ is requested, \ALG performs the following three operations,
in this order:
\begin{description}
    \item \emph{1. Target cleanup.} If $z^*$ was a target of
    another item $y$ (i.e.,~$\theta_y = z^*$ for $y \neq z^*$), then $\theta_y$~is
    updated to the successor of $z^*$. This happens for all items $y$ with  this property.
    \item \emph{2. Movement of $z^*$.} 
        \ALG executes one of the two actions: a \emph{partial move} or a \emph{full move}. 
        We will explain how to choose between them later.
        \begin{itemize}
            \item In the partial move, item $z^*$ is inserted right before $\theta_{z^*}$.  
			(If $\theta_{z^*}=z^*$, this means that $z^*$ does not change its position.)
            \item In the full move, item $z^*$ is moved to the front of the list.
        \end{itemize}    
    \item \emph{3. Target reset.} 
        $\theta_{z^*}$ is set to the front item of the list.
\end{description}
It is illustrative to note a few properties and corner cases of the algorithm. 
\begin{itemize}
    \item Target cleanup is executed only for items following $z^*$, and thus 
        the successor of~$z^*$ exists then (i.e., \ALG is well defined).
    \item If $\theta_{z^*} = z^*$ and a partial move is executed, then $z^*$ is not moved, but the items 
        that targeted $z^*$ now target the successor of $z^*$.
    \item For an item $x$, the items that precede $\theta_x$ in the list
        were requested (each at least once) after the last time $x$ had been requested.
\end{itemize}


\subparagraph*{Modes, flavors and states.}
Fix a pair $\set{x,y}$ such that $y \prec x$, and thus also $\theta_y \preceq y
\prec x$. This pair is assigned one of four possible \emph{flavors}, 
depending on the position of $\theta_x$ (cf.~\autoref{fig:flavors}):
\begin{itemize}
    \item flavor $d$ (\textbf{d}isjoint): if $\theta_y \preceq y \prec \theta_x \preceq x$,
    \item flavor $o$ (\textbf{o}verlapping): if $\theta_y \prec \theta_x \preceq y \prec x$,
    \item flavor $e$ (\textbf{e}qual): if $\theta_y = \theta_x \preceq y \prec x$,
    \item flavor $n$ (\textbf{n}ested): if $\theta_x \prec \theta_y \preceq y \prec x$ 
\end{itemize}

\begin{figure}[t]
\centering
\includegraphics[width=0.75\textwidth]{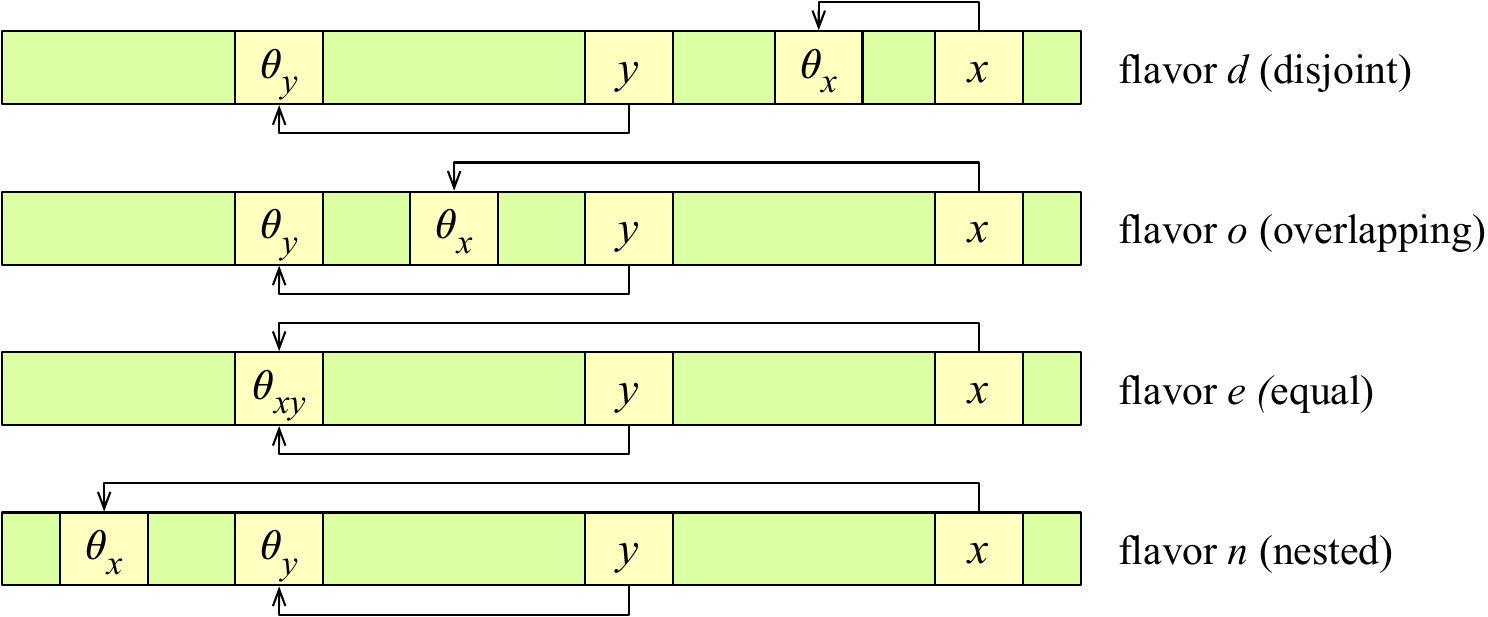}
\caption{The flavor of a pair $\set{x, y}$ (where $y \prec x$) depends on the position of
target $\theta_x$ with respect to items $\theta_y \preceq y \prec x$. In the figure for flavor $e$,
we write $\theta_{xy}$ for the item $\theta_x = \theta_y$.
}
\label{fig:flavors}
\end{figure}

For a pair $\set{x,y}$ that is in a mode $\xi \in \{ \alpha, \beta, \gamma
\}$ and has a flavor $\omega \in \{d,o,e,n\}$, we say that the \emph{state} of
$\set{x,y}$ is $\xi^\omega$. Recall that all pairs are initially in
mode $\alpha$, and note that their initial flavor is $d$. That is, the initial state of all pairs is
$\alpha^d$.

We will sometimes combine the flavors $o$ with $e$, stating that the pair is in
state $\xi^{oe}$ if its mode is $\xi$ and its flavor is $o$ or $e$.
Similarly, we will also combine flavors $n$ and $e$.


\subparagraph*{Pair potential.}
To each pair $\set{x,y}$ of items, we assign a non-negative \emph{pair potential}
$\Phi_{xy}$. We abuse the notation and use $\xi^\omega$ not only to denote
the pair state, but also the corresponding values of the pair potential. That
is, we assign the potential $\Phi_{xy} = \xi^\omega$ if pair $\set{x,y}$ is in
state~$\xi^\omega$. We pick the actual values of these potentials only later in
\autoref{sec:choosing_parameters}.

We emphasize that the states of each item pair depend on work functions for this
pair that are easily computable in online manner. Thus, \ALG may compute the
current values of~$\Phi_{xy}$, and also compute how their values would change for
particular choice of a move. 


\subparagraph*{Choosing the cheaper move.}
For a step $t$, let $\Delta_t \ALG$ be the cost of \ALG in this step, and 
for a pair $\set{x,y} \in \P$, let $\Delta_t \Phi_{xy}$ be the change of the
potential of pair $\set{x,y}$ in step~$t$. Let $z^*$ denote the requested
item. \ALG chooses the move (full or partial) with the smaller value of
\[
    \Delta_t \ALG + \sum_{y \prec z^*} \Delta_t \Phi_{z^*y},
\] 
breaking ties arbitrarily. Importantly,
note that the value that \ALG minimizes involves only pairs including the 
requested item $z^*$ and items currently preceding it.



\section{Analysis of Full-Or-Partial-Move}
\label{sec:analysis}


In this section, we show that for a suitable choice of pair potentials, 
\ALG is $3.3904$-competitive. 

To this end, we first make a few observations concerning how the modes, flavors
(and thus states) of item pairs change due to the particular actions of \ALG.
These are summarized in \autoref{tab:state_transitions_1},
\autoref{tab:state_transitions_2}, and \autoref{tab:state_transitions_3}, and
proved in \autoref{sec:mode_transitions} and \autoref{sec:state_transitions}.

Next, in \autoref{sec:amortized_analysis} and \autoref{sec:choosing_parameters},
we show how these changes influence the amortized cost of \ALG pertaining to particular
pairs. We show that for a suitable choice of pair potentials, 
we can directly compare the amortized cost of \ALG with the cost of \OPT.


\subsection{Structural Properties}

Note that requests to items other than $x$ and $y$ do not affect
the mutual relation between $x$ and $y$. Moreover, to some extent, they preserve
the relation between targets $\theta_x$ and~$\theta_y$,
as stated in the following lemma.

\begin{lemma}
\label{lem:theta_relation_preserved} 
    Fix a pair $\set{x,y} \in \P$ and assume that $\theta_y \preceq \theta_x$
    (resp.~$\theta_y = \theta_x$). If \ALG serves a~request to an~item~$z^*$
    different than $x$ and $y$, then these relations are preserved. The relation
    $\theta_y \prec \theta_x$ changes into $\theta_y = \theta_x$ only when
    $\theta_y$ and $\theta_x$ are adjacent and $z^* = \theta_y$.
\end{lemma}

\begin{proof}
    By the definition of \ALG, the targets of non-requested items are updated
    (during the target cleanup) only if they are equal to $z^*$. In such a case,
    they are updated to the successor of $z^*$. We consider several cases.
    \begin{itemize}
        \item If $z^* \notin \set{\theta_x, \theta_y}$, the targets $\theta_x$ and $\theta_y$ are not updated.
        \item If $z^* = \theta_y = \theta_x$, both targets are updated to the successor of $z^*$, and thus they 
            remain equal.
        \item If $z^* = \theta_y \prec \theta_x$, target $\theta_y$ is updated. If $\theta_y$ and 
            $\theta_x$ were adjacent ($\theta_y$ was an immediate 
            predecessor of $\theta_x$), they become equal.
            In either case, $\theta_y \preceq \theta_x$. 
        \item If $\theta_y \prec \theta_x = z^*$, target $\theta_x$ is updated, 
            in which case the relation $\theta_y \prec \theta_x$ is preserved.
        \qedhere
    \end{itemize}
\end{proof}


\subsection{Mode Transitions}
\label{sec:mode_transitions}
Now, we focus on the changes of modes. It is convenient to look first at
how they are affected by the request itself (which induces an update of the work
function), and subsequently due to the actions of \ALG (when some items are
swapped). The changes are summarized in the following observation.

\begin{observation}
\label{obs:type_change}
Let $z^*$ be the requested item. 
\begin{itemize}
    \item Fix an item $y \succ z^*$. The mode transitions of the pair 
        $\set{z^*,y}$ due to request are $\alpha \to \alpha$, $\beta \to \alpha$, and $\gamma \to
        \beta$. Subsequent movement of $z^*$ does not further change the mode.
    \item Fix an item $y \prec z^*$. Due to the request, pair $\set{z^*,y}$
        changes first its mode due to the request in the following way:
        $\alpha \to \beta$, $\beta \to \gamma$, and $\gamma \to \gamma$.
        Afterwards, if \ALG moves $z^*$ before $y$, the subsequent mode
        transitions for pair $\set{z^*,y}$ are $\beta \to \beta$ and $\gamma \to \alpha$.
\end{itemize}
\end{observation}


\subsection{State Transitions}
\label{sec:state_transitions}

Throughout this section, we fix a step and let~$z^*$ be the requested item. 

We analyze the potential changes for both types of movements. We split our considerations
into three cases corresponding to
three types of item pairs. The first two types involve $z^*$ as one pair item,
where the second item either initially precedes $z^*$ 
(cf.~\autoref{lem:state_transitions_1}) or follows $z^*$ 
(cf.~\autoref{lem:state_transitions_2}). The third type involves pairs 
that do not contain $z^*$ at all (cf.~\autoref{lem:state_transitions_3}).

While we defined 12 possible states (3 modes $\times$ 4 flavors), 
we will show that $\alpha^n$, $\gamma^d$, and~$\gamma^o$ never occur. 
This clearly holds at the very beginning as all pairs are then in
state~$\alpha^d$.

For succinctness, we also combine some of the remaining states, reducing the number of states
to the following six: 
$\alpha^d$, $\beta^d$, $\alpha^{oe}$, $\beta^o$, $\beta^{ne}$ and $\gamma^{ne}$.
(For example, a pair is in state $\alpha^{oe}$ if it is in state $\alpha^{o}$ or $\alpha^{e}$.)
In the following we analyze the transitions between them. We start with a
simple observation.


\begin{lemma}
    \label{lem:full_move}
    Fix an item $y \neq z^*$. If \ALG performs a full move, then 
    the resulting flavor of the pair~$\set{z^*,y}$ is $d$.
\end{lemma}

\begin{proof}
    Item $z^*$ is moved to the front of the list, and its target
    $\theta_{z^*}$ is reset to this item, i.e., $z^* = \theta_{z^*}$. After
    the movement, we have $z^* \prec \theta_y$. This relation follows trivially
    if $z^*$ is indeed moved. However, it holds also if $z^*$ was already on the
    first position: even if $\theta_y = z^*$ before the move, $\theta_y$ would be
    updated to the successor of $z^*$ during the target cleanup. The resulting
    ordering is thus $\theta_{z^*} = z^* \prec \theta_y \preceq y$, i.e., the
    flavor of the pair becomes $d$.
\end{proof}


\begin{lemma}
    \label{lem:state_transitions_1}
    Fix $y \prec z^*$. The state transitions for pair $\set{z^*,y}$
    are given in \autoref{tab:state_transitions_1}.
\end{lemma}

\begin{proof}
    First, suppose \ALG performs a full move. The pair $\set{z^*,y}$ is swapped,
    which changes its mode according to \autoref{obs:type_change} ($\alpha
    \to \beta$, $\beta \to \alpha$, $\gamma \to \alpha$). By
    \autoref{lem:full_move}, the flavor of the pair becomes $d$. This proves the
    correctness of the transitions in row three of \autoref{tab:state_transitions_1}.

    In the rest of the proof, we analyze the case when \ALG performs a partial
    move. We consider three sub-cases depending on the initial flavor of the
    pair $\set{z^*,y}$. For the analysis of mode changes we will apply
    \autoref{obs:type_change}.
    \begin{itemize}
    \item Before the movement, the flavor of the pair was $d$, i.e., $\theta_y
        \preceq y \prec \theta_{z^*} \preceq z^*$. 
        
        The movement of $z^*$ does not swap the pair, i.e., its mode
        transitions are $\alpha \to \beta$ and $\beta \to \gamma$. As
        $\theta_{z^*}$ is set to the list front, after the movement
        $\theta_{z^*} \preceq \theta_y \preceq y \prec z^*$, i.e., the flavor of
        the pair becomes either $e$ or~$n$. This explains the state transitions
        $\alpha^d \to \beta^{ne}$ and~$\beta^d \to \gamma^{ne}$.

    \item Before the movement, the flavor of the pair was $o$, i.e., 
        $\theta_y \prec \theta_{z^*} \preceq y \prec z^*$.

        Due to the movement, the pair is swapped, and its mode transitions
        are $\alpha \to \beta$ and $\beta \to \alpha$. As $\theta_{z^*}$ is set
        to the list front, we have $\theta_{z^*} \preceq \theta_y$. After the
        movement, $\theta_y \preceq z^* \prec y$, i.e., the flavor of the pair
        becomes either $o$ or $e$. This explains the state transitions
        $\alpha^{oe} \to (\beta^o$ or $\beta^{ne})$ and $\beta^o \to
        \alpha^{oe}$.

    \item Before the movement, the flavor of the pair was $e$ or $n$, i.e.,
        $\theta_{z^*} \preceq \theta_y \preceq y \prec z^*$.

        The movement swaps the pair, and its mode transitions are $\alpha
        \to \beta$, $\beta \to \alpha$ and~$\gamma \to \alpha$. After the
        movement $z^* \prec \theta_y$, and target $\theta_{z^*}$ is set to the
        list front, which results in $\theta_{z^*} \preceq z^* \prec \theta_y
        \preceq y$, i.e., the pair flavor becomes $d$. This explains the state
        transitions $\alpha^{oe} \to \beta^d$, $\beta^{ne} \to \alpha^d$, and
        $\gamma^{ne} \to \alpha^d$
    \qedhere
    \end{itemize}
\end{proof}

\begin{table}[t]
    \centering
    \begin{tabular}{|c|c|c|c|c|c|c|c|}
    \hline
    \textbf{State before move}
        & $\alpha^d$ & $\beta^d$ & $\alpha^{oe}$ & $\beta^o$ & $\beta^{ne}$ & $\gamma^{ne}$ \\
    \hline
    \textbf{State after partial move}
        & $\beta^{ne}$ & $\gamma^{ne}$ & $\beta^d$ or $\beta^o$ or $\beta^{ne}$ & 
            $\alpha^{oe}$ & $\alpha^d$ & $\alpha^d$ \\
    \hline
    \textbf{State after full move}
        & $\beta^d$ & $\alpha^d$ & $\beta^d$ & $\alpha^d$ & $\alpha^d$ & $\alpha^d$ \\
    \hline
    \end{tabular}
    \caption{State transitions for pairs $\set{z^*,y}$ where $y \prec z^*$ right before the request to  $z^*$.}
    \label{tab:state_transitions_1}
\end{table}


\begin{lemma}
    \label{lem:state_transitions_2}
    Fix $y \succ z^*$. The state transitions for pair $\set{z^*,y}$
    are given in \autoref{tab:state_transitions_2}.
\end{lemma}

\begin{table}[t]
    \centering
    \begin{tabular}{|c|c|c|c|c|c|c|c|}
    \hline
    \textbf{State before move}
        & $\alpha^d$ & $\beta^d$ & $\alpha^{oe}$ & $\beta^o$ & $\beta^{ne}$ & $\gamma^{ne}$ \\
    \hline
    \textbf{State after move}
        & $\alpha^d$ & $\alpha^d$ & $\alpha^d$ & $\alpha^d$ &
            $\alpha^d$ or $\alpha^{oe}$ & $\beta^d$ or $\beta^o$ or $\beta^{ne}$ \\ 
    \hline
    \end{tabular}
    \caption{State transitions for pairs $\set{z^*,y}$ where $z^* \prec y$ right before the request to $z^*$.}
    \label{tab:state_transitions_2}
\end{table}

\begin{proof}
    The pair $\set{z^*,y}$ is not swapped due to the request, and 
    thus, by \autoref{obs:type_change}, its mode transition is
    $\alpha \to \alpha$, $\beta \to \alpha$, $\gamma \to \beta$.

    If \ALG performs a full move, the flavor of the pair becomes $d$ by 
    \autoref{lem:full_move}. This explains the state transitions 
    $\alpha^d \to \alpha^d$, $\beta^d \to \alpha^d$, $\alpha^{oe} \to \alpha^d$, 
    $\beta^o \to \alpha^d$, $\beta^{ne} \to \alpha^d$, and $\gamma^{ne} \to \beta^d$.
    
    In the following, we assume that \ALG performs a partial move, 
    and we will identify cases where the resulting pair flavor is different than $d$. 
    We consider two cases.
    \begin{itemize}        
    \item The initial flavor is $d$, $o$ or $e$. That is 
        $\theta_{z^*} \preceq z^* \prec y$ and $\theta_{z^*} \preceq \theta_y \preceq y$. 
        During the target cleanup, $\theta_y$ may be updated to its successor, 
        but it does not affect these relations, and in particular we still have 
        $\theta_{z^*} \preceq \theta_y$. Thus, when $z^*$ is moved, 
        it gets placed before $\theta_y$. This results in the ordering 
        $\theta_{z^*} \preceq z^* \prec \theta_{y} \preceq y$, i.e., 
        the resulting flavor is $d$. 
    
    \item The initial flavor is $n$, i.e., 
        $\theta_y \prec \theta_{z^*} \preceq z^* \prec y$.
        As $\theta_y \neq z^*$, the target $\theta_y$ 
        is not updated during the target cleanup.
        As $z^*$
        is moved right before original position of $\theta_{z^*}$, it is placed
        after $\theta_y$, and the resulting ordering is $\theta_{z^*} \preceq
        \theta_y \prec z^* \prec y$. That is, the flavor becomes $o$ or~$e$,
        which explains the state transitions $\beta^{ne} \to \alpha^{oe}$ and
        $\gamma^{ne} \to (\beta^o$ or $\beta^{ne})$. 
    \qedhere
    \end{itemize}
\end{proof}

\begin{table}[t]
    \centering
    \begin{tabular}{|c|c|c|c|c|c|c|c|}
    \hline
    \textbf{State before move}
        & $\alpha^d$ & $\beta^d$ & $\alpha^{oe}$ & $\beta^o$ & $\beta^{ne}$ & $\gamma^{ne}$ \\
    \hline
    \textbf{State after move}
        & $\alpha^d$ & $\beta^d$ & $\alpha^{oe}$ & $\beta^o$ or $\beta^{ne}$ & $\beta^{ne}$ & $\gamma^{ne}$ \\ 
    \hline
    \end{tabular}
    \caption{State transitions for pairs $\set{x,y}$ where $x \neq z^*$ and $y \neq z^*$
    right before the request to $z^*$.}
    \label{tab:state_transitions_3}
\end{table}


\begin{lemma}
    \label{lem:state_transitions_3}
    Fix $y \prec x$, such that $x \neq z^*$ and $y \neq z^*$. 
    State transitions for pair $\set{x,y}$
    are given in \autoref{tab:state_transitions_3}.
\end{lemma}

\begin{proof}
    The mode of the pair $\set{x,y}$ is not affected by the request to
    $z^*$. 
    
    The flavor of the pair $\set{x,y}$ depends on mutual relations 
    between $x$, $y$, $\theta_x$ and $\theta_y$.
    By~\autoref{lem:theta_relation_preserved}, the only possible change is 
    that $\theta_x$ and $\theta_y$ were different but may become equal:
    this happens when they were adjacent and the earlier of them 
    is equal to $z^*$.
    We consider four cases depending on the initial flavor of the pair.    
    \begin{itemize}
        \item The initial flavor was $e$ ($\theta_y = \theta_x \preceq y \prec x$).
            As $\theta_y$ and $\theta_x$ are not adjacent, the flavor remains~$e$.
        \item The initial flavor was $d$ ($\theta_y \preceq y \prec \theta_x \preceq x$).
            Suppose $\theta_y = z^*$. As $y \neq z^*$, 
            we have $\theta_y = z^* \prec y \prec \theta_x$. That is, 
            $\theta_x$ and $\theta_y$ are not adjacent, and thus the flavor remains $d$.
        \item The initial flavor was $o$ ($\theta_y \prec \theta_x \preceq y \prec x$).
            In this case it is possible that $\theta_y$ and $\theta_x$ are adjacent and $z^* = \theta_y$.
            The flavor may thus change to $e$ or remain $o$.
        \item The initial flavor was $n$ ($\theta_x \prec \theta_y \preceq y \prec x$).
            Similarly to the previous case, it is possible that $\theta_x$ and $\theta_y$ are adjacent and $z^* = \theta_x$.
            The flavor may thus change to $e$ or remain~$n$.
        \qedhere
    \end{itemize}
\end{proof}


\subsection{Amortized Analysis}
\label{sec:amortized_analysis}

We set $R = \oneeighth (23+\sqrt{17}) \leq 3.3904$ as our desired competitive ratio.

In the following, we fix a step $t$ in which item $z^*$ is requested. We partition 
$\P$ into three sets corresponding to the three types of pairs:

\begin{itemize}\setlength{\itemsep}{0.025in}
    \item $\P^t_1 \triangleq \{ \{y, z^*\} : y \prec z^* \}$ (pairs where $z^*$ is the second item, 
    analyzed in \autoref{tab:state_transitions_1}),
    \item $\P^t_2 \triangleq \{ \{y, z^*\} : z^* \prec y \}$ (pairs where $z^*$ is the first item,
    analyzed in \autoref{tab:state_transitions_2}),
    \item $\P^t_3 \triangleq \{ \{x, y\} : x \neq z^* \wedge y \neq z^* \}$ (pairs where $z^*$ is not involved,
    analyzed in \autoref{tab:state_transitions_3}).
\end{itemize}

For succinctness, wherever it does not lead to ambiguity, we 
omit subscripts $t$, i.e., write $\Delta \Phi_{xy}$ and
$\Delta w^{xy}$ instead of $\Delta_t \Phi_{xy}$ and~$\Delta_t w^{xy}$. We also
omit superscripts $t$ in $\P^t_1$, $\P^t_2$, and~$\P^t_3$.

Our goal is to show the following three bounds:

\begin{itemize}\setlength{\itemsep}{0.025in}
    \item $\Delta \ALG + \sum_{ \set{x,y} \in \P_1} \Delta \Phi_{xy} 
        \leq R \cdot \sum_{ \set{x,y} \in \P_1} \Delta w^{xy}$, 
    \item $\sum_{ \set{x,y} \in \P_2} \Delta \Phi_{xy} 
        \leq R \cdot \sum_{ \set{x,y} \in \P_2} \Delta w^{xy}$, 
    \item $\sum_{ \set{x,y} \in \P_3} \Delta \Phi_{xy} 
        \leq R \cdot \sum_{ \set{x,y} \in \P_3} \Delta w^{xy}$.
\end{itemize}

Note that the left hand sides of these inequalities 
correspond to the portions of amortized cost of \ALG corresponding to sets $\P_1, \P_2, \P_3$,
while the right hand sides 
are equal to $R$ times the corresponding portion of the cost of pair-based \OPT.
Hence, if we can show the above inequalities for every step $t$, the competitive 
ratio of $R$ will follow by simply adding them up.

As we show in the sections that follow, these bounds reduce to some constraints involving state potentials
$\alpha^d$, $\beta^d$, $\alpha^{oe}$, $\beta^o$, $\beta^{ne}$ and $\gamma^{ne}$. 
The bounds for $\P_2$ and $\P_3$, while they involve summations over pairs, 
can be justified by considering individual pairs and the needed constraints are
simple inequalities between state potentials, summarized in the following assumption:

\begin{assumption}
\label{def:types_relations}
    We assume that
    \begin{itemize}
        \item $\alpha^d = 0$,
        \item all constants $\beta^d$, $\alpha^{oe}$, $\beta^o$, $\beta^{ne}$ and $\gamma^{ne}$ are non-negative,
        \item $\alpha^{oe} \leq \beta^{ne} + \half R \leq \beta^o + \half R$,
        \item $\max\{\beta^d, \beta^o \} \leq \gamma^{ne} + \half R$.
    \end{itemize}
\end{assumption}

The bound for $\P_1$ is most critical (not surprisingly, as it corresponds to requesting
the second item of a pair). To analyze this bound, the needed constraints, besides the state potentials 
also need to involve the numbers of pairs that are in these states.
This gives rise to a non-linear optimization problem that we need to solve.


\subsubsection{Analyzing pairs from set \texorpdfstring{$\P_1$}{P1}}

The proof of the following lemma is deferred to \autoref{sec:choosing_parameters}.

\begin{restatable}{lemma}{amortizedcost}
    \label{lem:amortized_cost_P1}
    There exist parameters 
    $\alpha^d$, $\beta^d$, $\alpha^{oe}$, $\beta^o$, $\beta^{ne}$ and $\gamma^{ne}$, 
    satisfying \autoref{def:types_relations},
    such that  
    for any step $t$, it holds that 
    $\Delta \ALG + \sum_{ \set{x,y} \in \P_1} \Delta \Phi_{xy} 
        \leq R \cdot \sum_{ \set{x,y} \in \P_1} \Delta w^{xy}$.
\end{restatable}


\subsubsection{Analyzing pairs from set \texorpdfstring{$\P_2$}{P2}}

\begin{lemma}
    \label{lem:amortized_cost_P2}
    For any step $t$, it holds that $\sum_{ \set{x,y} \in \P_2} \Delta \Phi_{xy} 
        \leq R \cdot \sum_{ \set{x,y} \in \P_2} \Delta w^{xy}$.
\end{lemma}

\begin{proof}
    Recall that $\P_2$ contains all pairs $\{z^*,y\}$, such that $z^* \prec y$.
    Thus, it is sufficient to show that 
    $\Delta \Phi_{z^*y} \leq R \cdot \Delta w^{z^*y}$ holds for any 
    such pair $\{z^*,y\}$. 
    The lemma will then follow by summing over all pairs from $\P_2$.

    By \autoref{def:types_relations}, 
    we have 
    \begin{align}
        \label{eq:obs_1}
        \alpha^d - \alpha^{oe} & \leq 0, \\
        \label{eq:obs_2}
        \max \{\alpha^d, \alpha^{oe} \} - \beta^{ne} 
            & = \alpha^{oe} - \beta^{ne} \leq \half R, \\
        \label{eq:obs_3}
        \max \{\beta^d, \beta^o, \beta^{ne}\} - \gamma^{ne} 
            & = \max \{ \beta^d, \beta^o \} - \gamma^{ne} \leq \half R.
    \end{align}
    
    We consider two cases depending on the initial mode of the pair $\set{z^*,y}$.
    In each case, we upper-bound the potential change on the basis of possible 
    state changes of this pair (cf.~\autoref{tab:state_transitions_2}).

    \begin{itemize}
    \item The initial mode of $\set{z^*,y}$ is $\alpha$. By
    \autoref{tab:state_transitions_2}, this mode remains $\alpha$, and thus by 
    \autoref{obs:type_change_causes_cost},
    $\Delta
    w^{z^*y} = 0$.    
    Then, 
    \begin{align*}
        \Delta \Phi_{z^*y} 
            & \leq \max \left\{ \alpha^d - \alpha^d, \alpha^d - \alpha^{oe} \right\} 
                && \text{(by \autoref{tab:state_transitions_2})} \\
            & = 0 = R \cdot \Delta w^{z^*y}.
                && \text{(by \eqref{eq:obs_1})} 
        \end{align*}

    \item The initial mode of $\set{z^*,y}$ is $\beta$ or $\gamma$. By
    \autoref{tab:state_transitions_2}, its mode changes due to the request
    to~$z^*$, and hence, by \autoref{obs:type_change_causes_cost}, $\Delta
    w^{z^*y} = \half$. Then, 
    \begin{align*}
        \Delta \Phi_{z^*y} 
            & \leq \max \{ 
                \alpha^d - \alpha^d, \alpha^d - \beta^d, \alpha^d - \alpha^{oe}, \alpha^d - \beta^o, \\
			&   \quad \quad \quad \quad   \max \{\alpha^d, \alpha^{oe} \} - \beta^{ne}, \\
            & \quad \quad \quad \quad \max \{\beta^d, \beta^o, \beta^{ne}\} - \gamma^{ne} \} 
                && \text{(by \autoref{tab:state_transitions_2})} \\
            & \leq \max \{ - \beta^d, -\beta^o, \half R, \half R \} 
                && \text{(by $\alpha^d = 0$, \eqref{eq:obs_1}, \eqref{eq:obs_2} and \eqref{eq:obs_3})} \\
            & = \half R = R \cdot \Delta w^{z^*y}.
    && \qedhere
    \end{align*}
    \end{itemize}
\end{proof}


\subsubsection{Analyzing pairs from set \texorpdfstring{$\P_3$}{P3}}

\begin{lemma}
    \label{lem:amortized_cost_P3}
    For any step $t$, it holds that $\sum_{ \set{x,y} \in \P_3} \Delta \Phi_{xy} 
    \leq R \cdot \sum_{ \set{x,y} \in \P_3} \Delta w^{xy}$.
\end{lemma}

\begin{proof}
    As in the previous lemma, we show that 
    the inequality
    $\Delta \Phi_{xy} \leq R \cdot \Delta w^{xy}$ holds for any 
    pair $\set{x,y} \in \P_3$, i.e., for a pair $\set{x,y}$, such that $x
    \neq z^*$ and $y \neq z^*$. The lemma will then follow by summing
    over all pairs $\set{x,y} \in \P_3$.

    Possible state transitions of such a pair $\set{x,y}$ are given in
    \autoref{tab:state_transitions_3}. Hence, such a pair either does not change
    its state (and then $\Delta \Phi_{xy} = 0$) or it changes it from $\beta^o$
    to $\beta^{ne}$ (and then $\Delta \Phi_{xy} = \beta^{ne} - \beta^o \leq 0$
    by \autoref{def:types_relations}). In either case, 
    $\Delta \Phi_{xy} \leq 0 \leq R \cdot \Delta w^{xy}$.
\end{proof}


\subsubsection{\texorpdfstring{Proof of $R$-competitiveness}{Proof of R-competitiveness}}

We now show that the three lemmas above imply that \ALG is $R$-competitive.

\begin{theorem}
    \label{thm:main_theorem}
    For an appropriate choice of parameters, the competitive 
    ratio of \ALG is at most $R = \oneeighth (23+\sqrt{17}) \leq 3.3904$.
\end{theorem}

\begin{proof}
Fix any sequence $\sigma$ consisting of $T$ requests. 
By summing the guarantees of \autoref{lem:amortized_cost_P1},
\autoref{lem:amortized_cost_P2}, and \autoref{lem:amortized_cost_P3}, we obtain
that for any step $t$, it holds that
\[
    \Delta_t \ALG + \sum_{ \set{x,y} \in \P} \Delta_t \Phi_{xy} 
        \leq R \cdot \sum_{ \set{x,y} \in \P} \Delta_t w^{xy}.
\]
By summing over all steps, observing that the potentials are non-negative and the
initial potential is zero
(cf.~\autoref{def:types_relations}), we immediately obtain that 
$
    \ALG(\sigma) \leq R \cdot \sum_{t=1}^T \sum_{ \set{x,y} \in \P} \Delta w^{xy}
\leq R \cdot \OPT(\sigma).
$
The second inequality follows by \autoref{lem:wf_vs_opt}.
\end{proof}


\subsection{Proof of \autoref{lem:amortized_cost_P1}}
\label{sec:choosing_parameters}

Again, we focus on a single step $t$, in which the requested item is denoted~$z^*$.
We let $A^d$, $B^d$, $A^{oe}$, $B^o$, $B^{ne}$, $C^{ne}$ be
the number of items $y$ preceding $z^*$ such that pairs $\set{z^*,y}$ have states
$\alpha^d$, $\beta^d$, $\alpha^{oe}$, $\beta^o$, $\beta^{ne}$ and $\gamma^{ne}$,
respectively. Let 
\[
    V \triangleq [A^d, B^d, A^{oe}, B^o, B^{ne}, C^{ne}].
\]
Note that $\|V\|_1$ is the number of items preceding $z^*$, and thus also the
access cost \ALG pays for the request. Moreover, $A^d + B^d$ is the number of
items that precede $\theta_{z^*}$.  We use~$\odot$ to denote scalar product
(point-wise multiplication) of two vectors. 

We define three row vectors $\G{OPT}$, $\G{PM}$, and $\G{FM}$, such that
\begin{align*}   
    \begin{pmatrix}
        \G{PM} \\
        \G{FM} \\
        \G{OPT}
    \end{pmatrix} 
    & = 
    \begin{pmatrix}
        1 & 1 & 2 & 2 & 2 & 2 \\
        2 & 2 & 2 & 2 & 2 & 2 \\
        \half & \half & \half & \half & \half & 0 
    \end{pmatrix} \\
    & + 
    \begin{pmatrix}
        \beta^{ne} & (\gamma^{ne} - \beta^d) & (\max\{\beta^d, \beta^o\} - \alpha^{oe}) & (\alpha^{oe} - \beta^o) & -\beta^{ne} & -\gamma^{ne} \\
        \beta^d & -\beta^d & (\beta^d - \alpha^{oe}) & -\beta^o & -\beta^{ne} & -\gamma^{ne} \\
        0 & 0 & 0 & 0 & 0 & 0
    \end{pmatrix}
\end{align*}


\subparagraph*{Expressing costs as vector products.}
Recall that to prove \autoref{lem:amortized_cost_P1}, we need to relate the $\P_1$ portion of
the amortized cost of \ALG, i.e., $\Delta \ALG + \sum_{ \set{x,y} \in \P_1}
\Delta \Phi_{xy}$ and the corresponding portion of the cost of pair-based \OPT, i.e., $\sum_{ \set{x,y} \in
\P_1} \Delta w^{xy}$. In the following two lemmas, we show how to express both
terms as vector products.


\begin{lemma}
\label{lem:wf_as_product}
    It holds that $\sum_{y \prec z^*} \Delta w^{z^*y} = \G{OPT} \odot V$.
\end{lemma}

\begin{proof}
    The right-hand side of the lemma relation is equal to 
    $\half ( A^d + B^d + A^{oe} + B^o + B^{ne} )$.
    By \autoref{obs:type_change}, due to request to $z^*$: 
    \begin{itemize}
        \item $A^d + A^{oe}$ pairs change mode from $\alpha$ to $\beta$, and 
        \item $B^d + B^o + B^{ne}$ pairs change mode from $\beta$ to $\gamma$.
    \end{itemize}
    By \autoref{obs:type_change_causes_cost}, each such mode change contributes 
    $\half$ to the left hand side of the lemma equation. Note that 
    $C^{ne}$ pairs of mode $\gamma$ do not change their mode due to the request.
\end{proof}


\begin{lemma}
\label{lem:potential_change_as_product}  
    It holds that $\Delta \ALG + \sum_{y \prec z^*} \Delta \Phi_{z^*y} = G \odot V$,
    where 
    \begin{itemize} 
        \item $G = \G{PM}$ if \ALG performs a partial move, and 
        \item $G = \G{FM}$ if \ALG performs a full move.
    \end{itemize}
\end{lemma}

\begin{proof}
First, assume that \ALG performs a partial move. 
By the definitions 
of $A^d$, $B^d$, $A^{oe}$, $B^o$, $B^{ne}$, and $C^{ne}$,
$\Delta \ALG = [1,1,2,2,2,2] \odot V$.
By \autoref{tab:state_transitions_1}, 
\begin{align*}
    \sum_{y \prec z^*} & \Delta \Phi_{z^*y} \\
        & \leq \big[\beta^{ne} - \alpha^d, 
            \gamma^{ne} - \beta^d, 
            \max\{\beta^d, \beta^o, \beta^{ne}\} - \alpha^{oe}, 
            \alpha^{oe} - \beta^o, 
            \alpha^d - \beta^{ne}, 
            \alpha^d - \gamma^{ne}
        \big] \odot V \\
        & = \left[\beta^{ne}, 
            \gamma^{ne} - \beta^d, 
            \max\{\beta^d, \beta^o\} - \alpha^{oe},
            \alpha^{oe} - \beta^o,
            - \beta^{ne},
            - \gamma^{ne}
        \right] \odot V ,
\end{align*}
where the second equality follows as
$\alpha^d = 0$ and $\beta^o \geq \beta^{ne}$ 
(by \autoref{def:types_relations}).
Thus, the lemma holds for a partial move.

Next, assume \ALG performs a full move.  Then, $\Delta \ALG = [2,2,2,2,2,2] \odot V$.
By \autoref{tab:state_transitions_1},
\begin{align*}
    \sum_{y \prec z^*} \Delta \Phi_{z^*y} 
        & = \left[
            \beta^d - \alpha^d,
            \alpha^d - \beta^d, 
            \beta^d - \alpha^{oe}, 
            \alpha^d - \beta^o, 
            \alpha^d - \beta^{ne},
            \alpha^d - \gamma^{ne}
            \right] 
        \odot V \\
        & = \left[
            \beta^d, 
            -\beta^d, 
            \beta^d - \alpha^{oe}, 
            -\beta^o,
            -\beta^{ne}, 
            -\gamma^{ne}
            \right] 
        \odot V,
\end{align*}
where in the second equality we used $\alpha^d = 0$ (by \autoref{def:types_relations}).
Thus, the lemma holds for a~full move as well.
\end{proof}

Recall that \ALG is defined to choose the move that minimizes 
$\Delta \ALG + \sum_{y \prec z^*} \Delta \Phi_{z^*y}$.

\begin{corollary}
\label{cor:potential_change_as_product}
    It holds that $\Delta \ALG + \sum_{y \prec z^*} \Delta \Phi_{z^*y} = 
    \min \{ \G{PM} \odot V, \G{FM} \odot V\}$.
\end{corollary}


\subparagraph*{Finding Parameters.}

We may now prove \autoref{lem:amortized_cost_P1},
restated below for convenience.

\amortizedcost*

\begin{proof}
    We choose the following values of the parameters:
    \begin{align*}
    \alpha^d &= 0  & \beta^d &= \onesixteenth (5+3\sqrt{17})\approx 1.086 & \alpha^{oe} &= 2  
    \\
    \beta^o &= \onesixteenth (1+7\sqrt{17}) \approx 1.866 & \beta^{ne} &=  \onesixteenth(9-\sqrt{17})  \approx 0.305  &  \gamma^{ne} &= 2
    \end{align*}
    It is straightforward to verify that these values 
    satisfy the conditions of \autoref{def:types_relations}.
    We note that relation $\alpha^{oe} \leq \beta^{ne} + \half R$ holds with equality.

    By \autoref{cor:potential_change_as_product},
    $\Delta \ALG + \sum_{ \set{x,y} \in \P_1} \Delta \Phi_{xy} 
    = \Delta \ALG + \sum_{y \prec z^*} \Delta \Phi_{z^*y} = 
    \min \{ \G{PM} \odot V, \G{FM} \odot V\}$.
    On the other hand, by \autoref{lem:wf_as_product},
    $\sum_{ \set{x,y} \in \P_1} \Delta w^{xy} = 
    \sum_{y \prec z^*} \Delta w^{z^*y} = \G{OPT} \odot V$.
    Hence, it remains to show that 
    $\min \{ \G{PM} \odot V, \G{FM} \odot V\} \leq R \cdot \G{OPT} \odot V$.
    
    We observe that
    \begin{alignat*}{10}
        \G{PM} & = 
            \onesixteenth \cdot \big[
                25-\sqrt{17},\; 
                && 43-3\sqrt{17}, \; 
                && 1+7\sqrt{17},\; 
                && 63-7\sqrt{17},\; 
                && 23+\sqrt{17},\;
                && 0
            \big], \\
        \G{FM} & = 
            \onesixteenth \cdot \big[
               37+3\sqrt{17},\; 
               && 27-3\sqrt{17},\; 
               && 5+3\sqrt{17},\; 
               && 31-7\sqrt{17},\;
               && 23+\sqrt{17},\;
               && 0
            \big]. 
    \end{alignat*}
    Let $c = \frac{1}{4} (\sqrt{17}-1)$ and let $\G{COMB} = c \cdot \G{PM} + (1-c) \cdot \G{FM}$.
    Then, 
    \begin{align*}
        \G{COMB} = 
            \onesixteenth \cdot \big[
                23+\sqrt{17},\; 
                23+\sqrt{17}, \; 
                23+\sqrt{17},\; 
                23+\sqrt{17},\; 
                23+\sqrt{17},\;
                0
            \big].
    \end{align*}
    Now for any vector $V$, 
    \begin{align*}
        \min \{ \G{PM} \odot V, \G{FM} \odot V \} 
        & \leq c \cdot \G{PM} \odot V + (1-c) \cdot \G{FM} \odot V \\
        & = (c \cdot \G{PM} + (1-c) \cdot \G{FM}) \odot V \\
        & = \G{COMB} \odot V 
         = R \cdot \G{OPT} \odot V,
    \end{align*}
\noindent
completing the proof.\qedhere
\end{proof}


\section{Final Remarks}
\label{sec:discussion}

The most intriguing question left open in our work is whether competitive ratio
of $3$ can be achieved. We have shown computationally
(see~\autoref{sec:computational}) that $3$-competitive algorithms exist for
lists with up to $6$ items. 

However, even for short lists the definition of such $3$-competitive algorithm
remains elusive. For many online problems, the most natural candidate is the
generic work function algorithm. This algorithm is $2$-competitive in the
$\probLUPstandard$ model~\cite{1999_anderson_lup_work_function}. However, our
computer-aided calculation of its performance shows that its ratio is larger than
$3$ already for $5$ items (see~\autoref{sec:computational}).
It is $3$-competitive for lists of length up to $4$, though.

We do not know whether the analysis of $\ALG$ is tight. 
For the specific
choice of parameters used in the paper, we verified that $\ALG$ is
$3$-competitive for lists of length $3$ (see~\autoref{sec:length-3}), but not
better than $3.04$-competitive for lists of length $5$
(see~\autoref{sec:lower-fpm}). 

The focus of this paper is on the $\probLUPuniform$ model (also known as $P^1$);
we believe that the setting of $d = 1$ captures the essence and hardness 
of the deterministic variant. That said, extending the definition and analysis of $\ALG$ to the
$P^d$ model (for arbitrary $d$) is an~interesting open problem that 
deserves further investigation. 


\bibliography{list_update_references.bib}


\appendix

\section{Handling insertions and deletions}
\label{sec:dynamic}

In the ``dynamic'' variant of List Update, 
except from access requests, the input $\sigma$ may contain
insertions and deletions of items. By the variant
definition (see, e.g.,~\cite{2008_albers_lup_EfoA}): 
\begin{itemize}
\item inserting $z^*$ places it at the end of the list, and its cost is defined
    to be equal to the current length of the list;
\item deleting $z^*$ involves accessing $z^*$ (at a cost equal to its position),
    and then removing it.
\end{itemize}
Right after the insertion, an algorithm may reorganize the list 
paying the usual cost of swaps. Similarly, such reorganization
is possible right after accessing the item to be deleted, 
and before its actual removal.

In this section, we demonstrate that \ALG can handle insertions and deletions
while retaining its competitive ratio.
We start with a few definitions and clarifications, showing how the arguments
for lower-bounding \OPT need to be adapted for the dynamic variant. Later, we
we specify what swaps \ALG performs due to insertions and deletions.

\subparagraph{Adjusted cost.}

Let \emph{adjusted cost} of an algorithm be the cost
that omits the cost of serving insertions of items. (It does involve the cost of
an optional list reorganization). As these costs are 
the same for both \ALG and \OPT, if we can show that an algorithm is
$R$-competitive with respect to the adjusted cost, then it is also
$R$-competitive with respect to the actual cost. From this point on, we use only
adjusted cost in our analysis.

\subparagraph{Lifetimes of pairs.}

Without loss of generality, we assume that when an item is deleted, it is never
inserted back to the list (we simply rename its next occurrences). In the
main part of the paper, we used $\P$ to denote the set of all pairs of items.
For the dynamic variant, we will use $\Q$ to denote the set of all pairs of
items that are ever together in the list.
We call a~pair $\set{x,y} \in \Q$
\begin{itemize}
\item \emph{dormant} before $x$ and $y$ are both present in the list;
\item \emph{active} when both $x$ and $y$ are present in the list;
\item \emph{inactive} after either $x$ or $y$ is deleted from the list.
\end{itemize}
Each pair of $\Q$ is active within a certain period. This period may be preceded by
the dormant one, and may be followed by the inactive one. (In the static case,
all pairs are always active.)

\subparagraph{Pair-based OPT.}

Previously, we related both \ALG and \OPT to optimal algorithms for sequences
$\sigma_{xy}$ operating on two-item lists containing $x$ and $y$. In the
dynamic variant, we define $\sigma_{xy}$ to contain requests only from the
period when the pair $\set{x,y}$ is active. The corresponding two-item list
is also defined only within this period. The initial order of items on
this two-item list matches the order of $x$ and~$y$ on the actual
list at the beginning of the active period. This implies that each pair starts its 
active period in mode $\alpha$.

With these definitions in place, \autoref{lem:pairwise_opt_vs_opt} extends in
a~straightforward way to the dynamic case, i.e., $\sum_{\set{x,y} \in \Q}
\OPT(\sigma_{xy}) \leq \OPT(\sigma)$ for each input sequence $\sigma$. This also
implies a~counterpart of \autoref{lem:wf_vs_opt} for the dynamic variant: for
any input sequence $\sigma$ consisting of $T$ requests, it holds that
$\sum_{t=1}^T \sum_{\set{x,y} \in \Q} \Delta_t w^{xy} \leq \OPT(\sigma)$.

\subparagraph{Potential function.}

Recall that the potential function $\Phi_{xy}$ for a pair $\set{x,y}$ depends on
the corresponding work function $W^{xy}$, which is defined only for active
pairs. Therefore, for completeness, we set the potential function for a dormant
pair to be zero, and assume that the value of the potential function is frozen
once a pair stops being active and becomes inactive.

\subparagraph{Specifying the behavior of the algorithm.}

The behavior of \ALG for access requests remains the
same as in the static case. 
For insertions, \ALG performs no list reorganization. 
Finally, suppose that an element $z^*$ is deleted. 
Upon such request, \ALG performs exactly the same list reorganization as if 
$z^*$ was accessed (i.e., it moves $z^*$ fully or partially towards 
the front of the list). Once this action is performed, due 
to the target cleanup operation, $z^*$~is not a~target for any other item in the list.

Admittedly, the behavior of \ALG for deletions is wasteful
as it moves the element which is about to be removed from the list.
However, this definition streamlines the analysis, as we can 
reuse the arguments for access requests.

\subparagraph{Competitive ratio.}

We first argue about potential changes due to insertions and deletions.
Later, we use these arguments to show that \ALG retains 
its competitive ratio also in the dynamic variant.

We split the action of deleting an item $z^*$ into two parts: accessing the item (with
induced list reorganization) and \emph{actual removal} ot this item from the
list. 

\begin{lemma}
    \label{lem:insertion_potential} 
    Fix a pair $\set{x,y} \in \Q$. Due to an insertion or an actual removal,
    $\Delta \Phi_{xy} = 0$.
\end{lemma}

\begin{proof}
    Let $z^*$ be the item either inserted or actually removed; 
    we consider two cases.
    
    In the first case, pair $\set{x,y}$ does not contain $z^*$.
    If pair $\set{x,y}$ remains dormant or inactive,
    then $\Phi_{xy}$ is trivially unaffected. Otherwise, 
    the pair must be and remain active. 
    It suffices to note that 
    neither the insertion of $z^*$ nor the actual removal of $z^*$ 
    does affect $\theta_x$ or $\theta_y$.
    Thus, the mode and flavor of $\set{x,y}$ remains unchanged, and 
    hence $\Delta \Phi_{xy} = 0$.
    
    In the second case, one of pair elements, say $x$, is equal to $z^*$.
    \begin{itemize}
    \item For insertions, this means that pair $\set{z^*,y}$ 
    switches from the dormant state (of zero potential)
    to the active one. 
    As noted above, $\set{z^*,y}$ starts in mode $\alpha$.
    Since \ALG sets $\theta_{z^*} = z^*$, it holds that 
    $\theta_y \preceq y \prec \theta_{z^*} = z^*$, i.e., pair $\set{z^*,y}$ 
    starts in flavor $d$.
    Hence, the state of $\set{z^*,y}$ right after the insertion 
    is $\alpha^d$ which corresponds to the zero potential. 
    Thus, $\Delta \Phi_{xy} = 0$.
    \item For actual removals, pair $\set{z^*,y}$ switches from 
    the active state to the inactive one. 
    The pair potential becomes frozen, and thus $\Delta \Phi_{xy} = 0$.
    \qedhere
    \end{itemize}
\end{proof}

\begin{theorem}
\label{thm:main_theorem_dynamic}
    The competitive ratio of \ALG in the dynamic variant is 
    at most $R = \oneeighth (23+\sqrt{17}) \leq 3.3904$.
\end{theorem}

\begin{proof}
Fix any sequence $\sigma$ consisting of $T$ requests.
By \autoref{thm:main_theorem}, for any step $t$ containing 
an access request, we have that
\begin{equation}
    \label{eq:main_theorem_dynamic}
        \Delta_t \ALG + \sum_{ \set{x,y} \in \Q} \Delta_t \Phi_{xy} 
            \leq R \cdot \sum_{ \set{x,y} \in \Q} \Delta_t w^{xy}.
\end{equation}
We observe that the initial potentials for all pairs of $\Q$ (including the
dormant ones) are zero. Thus, to complete the proof, it remains to show that
\eqref{eq:main_theorem_dynamic} holds for insertions and deletions. The
competitive ratio of \ALG then follows by summing
\eqref{eq:main_theorem_dynamic} over all steps $t$, using the non-negativity of pair
potentials, and finally applying the relation $\sum_{t=1}^T \sum_{ \set{x,y} \in
\Q} \Delta w^{xy} \leq \OPT(\sigma)$ as in the proof of
\autoref{thm:main_theorem}.

For insertions, by the definition of adjusted cost, we have that $\Delta_t
\ALG = \Delta_t \OPT = 0$. By \autoref{lem:insertion_potential}, 
$\Delta_t \Phi_{xy} = 0$ for any pair $\set{x,y} \in \Q$,
and thus \eqref{eq:main_theorem_dynamic} holds.

For deletions, we can split the action into two parts: 
accessing the item (and the resulting list reorganization) 
and actual removal of the item.
For the first part, \eqref{eq:main_theorem_dynamic} holds
by \autoref{thm:main_theorem}. For the actual removal,
$\Delta_t \ALG = \Delta_t \OPT = 0$.
Furthermore, \autoref{lem:insertion_potential} implies that 
$\Delta_t \Phi_{xy} = 0$ for any pair $\set{x,y} \in \Q$,
and thus \eqref{eq:main_theorem_dynamic} holds.
\end{proof}

\section{Lower Bounds on MTF and its Modifications}
\label{sec:lower}

In this section, we show the lower bounds greater than $3$ 
for several natural algorithms, most of which have previously been considered for this problem. 
All these bounds (unless explicitly stated) hold for the full cost model, and thus also for the partial one.


\subsection{Specific Algorithms}

\subparagraph*{Deterministic \BIT.}

Algorithm Deterministic \BIT (\DBIT) starts with all items unmarked. 
On a request to item $r$, if it is marked, it moves $r$ to 
the front and unmarks it. Otherwise, it marks $r$.
We show that the competitive ratio of \DBIT is at least $3.302$.

To this end, consider a list of length $n$ and denote its items as $x_0, x_1,
\ldots, x_{n-1}$, counting from the front of the list. We divide the list into
two parts, $A$ and $B$, where part $A$ has length $cn$ and part $B$ has
length $(1-c)n$. The value of parameter $c$ will be determined later. Formally
$A = (x_0, x_1, x_2, \ldots, x_{cn-1})$ and $B = (x_{cn}, x_{cn+1}, \ldots,
x_{n-2}, x_{n-1} )$. For succinctness, we denote such order of the list as
$A, B$. We define input
\begin{align*}
    \sigma = (  &x_{cn-1}, x_{cn-2}, \ldots, x_1, x_0,\\
                &x_{n-1}, x_{n-1}, x_{n-2}, x_{n-2}, x_{n-3}, x_{n-3}, \ldots, x_{cn}, x_{cn} )^2.
\end{align*}
In this sequence, all items of $A$ are requested first, starting from the one
furthest in the list. Then every item of $B$ is requested twice, also starting
from the furthest one. It results in items of $B$ being transported to the
front, thus after the first half of $\sigma$, the order of the list of \DBIT is
$B,A$. Note that order of items within $B$ does not change.

Then the same sequence of requests is repeated. The difference is that now
requesting any item of $A$ results in moving that item to the front. Then
items of $B$ are again requested twice, resulting in final list of algorithm
being $B, A$. Then, 
\begin{align*}
\DBIT(\sigma) 
    & = \frac{cn(cn+1)}{2} && \text{(access to $A$)} \\
    & \quad + (1-c)n \cdot (3n-1) && \text{(access to $B$ and movement)}\\
    & \quad + cn \cdot (2n-1) && \text{(access to $A$ and movement)}\\
    & \quad + (1-c)n \cdot (3n-1) && \text{(access to $B$ and movement)} \\
    & = n^2 \left(6 - 4c + \frac{c^2}{2} \right) + n \left( \frac{3c}{2} - 2 \right).
\end{align*}

Now we define a better algorithm for this sequence, which upper-bounds the cost
of \OPT. After receiving the first batch of
requests to the items of $A$, it moves all the items of $B$ before the items of
$A$. It does not performs any further swaps. Thus, 
\begin{align*}
\OPT(\sigma) 
    & \leq \frac{cn(cn+1)}{2} && \text{(access to $A$)} \\
    & \quad + cn \cdot (1-c)n  && \text{(swap $A$ with $B$)}\\
    & \quad + 2 \cdot \frac{1+(1-c)n}{2} \cdot (1-c)n && \text{(access to $B$)}\\
    & \quad + \frac{(1-c)n+1+n}{2} \cdot cn  && \text{(access to $A$)} \\
    & \quad + 2 \cdot \frac{1+(1-c)n}{2} \cdot (1-c)n && \text{(access to $B$)}\\
    & = n^2 \left(2 - 2c + c^2 \right) + n \left( 2-c\right).
\end{align*}
Note that the final order of the list in both algorithms is the same and
equal to $B, A$. Moreover, every item is requested an even number of times,
so no item is marked after $\sigma$. Thus, this request sequence can be
repeated infinitely many times, with sets $A, B$ redefined after each
sequence.

With a long enough list, $n$ becomes insignificant compared to $n^2$,
and therefore the ratio is $\DBIT(\sigma) / \OPT(\sigma) \geq (6 -
4c + \frac{1}{2}c^2) / (2 - 2c + c^2)$. For $c = \frac{1}{3} (5 - \sqrt{13})$, we have 
$\DBIT(\sigma) / \OPT(\sigma) \geq \frac{1}{2} (3 + \sqrt{13}) \geq
3.302$.


\subparagraph*{Deterministic \BIT in partial cost model.} 

We now show that the competitive ratio of \DBIT is at least $4$
in the partial cost model. We 
recursively construct a sequence that achieves this ratio for every length of
the list $n \ge 2$. For a list of length $n$, the created sequence will 
be denoted $\sigma_n$. For reasons that will become apparent soon, we index items
from the end of the list, namely the algorithm's list at the start of the
sequence is $x_{n-1}, x_{n-2}, \ldots, x_1, x_0$.

First, assume $n = 2$. Set
\begin{equation*}
    \sigma_2 = x_1, x_0, x_0, x_1, x_0, x_0.
\end{equation*}
Clearly, $\DBIT(\sigma_2) = 8$. The cost of $2$ can be achieved for that sequence by
starting with the list order $x_1, x_0$ and swapping items only after the
first request. Note that such an algorithm would have the same order of items 
as algorithm \DBIT after serving $\sigma_2$.

Now let $n > 2$. Set 
\begin{equation*}
\sigma_n \;=\; x_{n-1}, \sigma_{n-1}, x_{n-1}, \sigma'_{n-1}.
\end{equation*}
where by $\sigma'_{n-1}$ above we mean the sequence obtained from $\sigma_{n-1}$
by permuting the items $x_{n-2},\dots,x_{0}$ according to their position in the
list after executing $\sigma_{n-1}$. Specifically, if $x_{\pi(n-2)}, x_{\pi(n-3)},
\dots, x_{\pi(0)}$ is the list after executing $\sigma_{n-1}$ then 
$\sigma'_{n-1}$ is obtained from $\sigma_{n-1}$ by replacing each item $x_i$ by
$x_{\pi(i)}$. For example, for $n=3$, 
\begin{equation*}
    \sigma_3 \;=\; x_2 , x_1 ,  x_0  , x_0 , x_1 ,  x_0  , x_0 , x_2, x_0 ,  x_1  , x_1 , x_0 ,  x_1  , x_1 .
\end{equation*}

We now express $\DBIT(\sigma_n)$ and $\OPT(\sigma_n)$ as functions of 
$\DBIT(\sigma_{n-1})$ and $\OPT(\sigma_{n-1})$, respectively.

At the start of $\sigma_n$, item $x_{n-1}$ is at the front, marked because of
the first request. Within $\sigma_{n-1}$ every item other than $x_{n-1}$ is
requested even number of times (at least twice). For the first pair of such
requests, $x_{n-1}$ contributes $3$ to the cost (one per each request and one
for the swap when the requested item is moved to the front). Thus, the
additional cost due to $x_{n-1}$ is at least $3(n-1)$. Then, item $x_{n-1}$ is
requested and moved to the front incurring cost $2(n-1)$. Similarly, $x_{n-1}$
contributes an additional cost of $3(n-1)$ to $\sigma'_{n-1}$. Thus,
\begin{equation}\label{eq:dbit1}
\DBIT(\sigma_n) \;\ge\; 2\cdot \DBIT(\sigma_{n-1}) + 8\cdot (n-1).
\end{equation}

The optimal algorithm can serve the first request to $x_{n-1}$ at cost $0$ and then move it
immediately to the end, paying $2(n-1)$ for the second request to $x_{n-1}$. 
In the final list of \OPT, $x_{n-1}$ is at the end, just like in the list 
of \DBIT. This implies that
serving $\sigma_{n-1}$ and $\sigma'_{n-1}$ will cost at most $2 \cdot
\OPT(\sigma_{n-1})$, i.e.,
\begin{equation}\label{eq:dbit2}
\OPT(\sigma_n) \;\le\; 2\cdot  \OPT(\sigma_{n-1}) + 2\cdot (n-1).
\end{equation}
Recurrence relations \eqref{eq:dbit1} and \eqref{eq:dbit2}  immediately imply that $\DBIT(\sigma_n)/\OPT(\sigma_n) \ge 4$. 


\subparagraph*{Half-Move.}
Finally, we consider the algorithm Half-Move, which moves the requested item to
the middle point between its position and the front of the list. We show that
its  competitive ratio  is at least $6$.

Consider a list of even length $n$. Denote the items as $x_0, x_1, \ldots,
x_{n-1}$, according to their positions in the initial list of the algorithm. The
request sequence is defined as
\[ 
    \sigma = \left( x_{n-1}, x_{n-2}, x_{n-3}, \dots , x_{n/2} \right)^k
\]
for some $k \geq 1$. At
the time of the request, the requested item is at the end of the list and it is
moved to the middle of the list. Since in the list of even length there are two
middle positions, we will assume that the item is always moved to the one
further from the front of the list. (If it is moved to the other one, the proof
only needs to be slightly modified.) The cost of each request for the
algorithm is $n + \frac{1}{2}n - 1 = \frac{3}{2}n - 1$. Moreover, after
$\frac{1}{2}n$ requests, the list returns to its original order. Therefore, the
total cost of the sequence is $k \cdot \frac{1}{2}n \cdot (\frac{3}{2}n - 1)$.

It is possible to change the order of the list to $x_{n/2}, x_{n/2+1}, \ldots,
x_{n-1}, x_0, x_1, x_2, \ldots, x_{n/2-1}$ before any requests arrive, at cost
less than $n^2$. Thus, the total cost of the request
sequence is $k \cdot \binom{n/2}{2}$. Thus, $\OPT(\sigma)
\leq n^2 + \frac{1}{8}kn^2$. The ratio approaches $6$ for
sufficiently large $n$, as $k \rightarrow \infty$.


\subsection{Stay-or-MTF Class}
\label{subsec:stay-or-mtf}

We now consider a class of online algorithms that never make partial moves. That
is, when an item $x$ is requested, an algorithm either does not change the
position of $x$ or moves $x$ to the front. We refer to such
algorithms as \emph{Stay-or-MTF} algorithms.

We show that if $A$ is a Stay-or-MTF algorithm, then the competitive ratio of
$A$ is at least~$3.25$, even for a list of length $3$.

The adversary strategy is illustrated in the figure below. The list has three
items named $a$, $b$ and $c$. The vertices represent states of the game. In each
state, we assume that the $A$'s list is $abc$. When $A$ executes a
move, the items are appropriately renamed. Each state is specified by the
current \emph{offset function}, equal to the current
work function minus the minimum of this work function. At each move, the
increase of the minimum of the work function is charged towards the \OPT's 
cost, and the offset function is updated. Each offset function that appears in
this game is specified by the set of permutations of $a,b,c$ where its value is
$0$. (Its value at each other list is equal to the minimum swap distance to some
offset function with value $0$.) To specify this set, we indicate, using braces,
which items are allowed to be swapped. For example, $a\braced{bc}$ represents
two permutations, $abc$ and $acb$, and $\braced{abc}$ represents all
permutations.

The adversary strategy is represented by the transitions between the states. For
each state the adversary specifies the generated request when this state is
reached. Notation $r/\xi$ on an arrow means that the request is $x$ and the
adversary cost is $\xi$. Then, the algorithm decides to make a move. The
decision points of the algorithm are indicated by small circles. The type of
move is uniquely specified by the cost, specified on the edges from these
circles. For example, if the request is on $c$, the algorithm either stays,
which costs $2$, or moves to front, which costs $4$.

\begin{center}
    \includegraphics[width=3.5in]{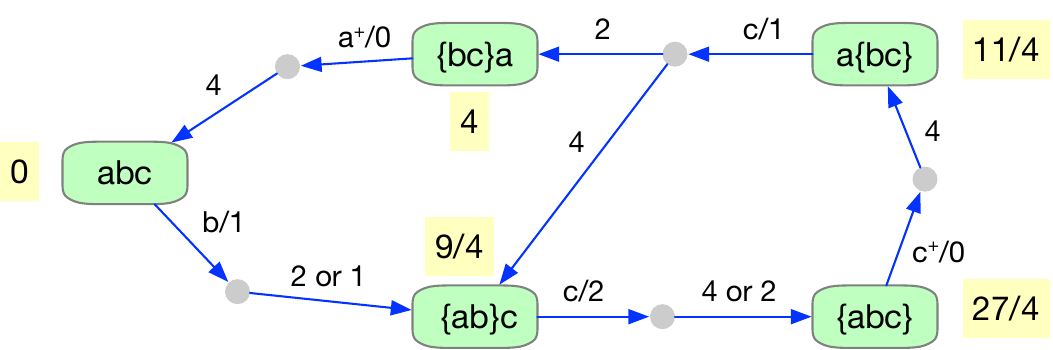}
\end{center}

The numbers next to each state are the values of potential $\Psi$.
By routine verification of the transition graph, 
we can check that for each move we have
\begin{equation*}
    \Delta A +\Delta\Psi \; \ge\; \thirteenfourths \cdot \Delta\OPT,
\end{equation*}
where $ \Delta A$, $\Delta\Psi$ and $\Delta\OPT$ represent the cost of $A$, the potential
change, and the optimum cost for this move.
This shows that the competitive ratio of $A$ is at least $\thirteenfourths = 3.25$.

Our construction is tight for $n = 3$, i.e., it can be shown that there exists a
Stay-or-MTF algorithm with the competitive ratio of $3.25$ for $n = 3$. For $n=4$,
we can improve the lower bound to $3.3$, using an approach that is essentially 
the same, but requires many more states.

\section{A 3-Competitive Algorithm for Lists of Length 3}
\label{sec:length-3}


The algorithm \ALG is $3$-competitive for lists of length 3. This is true even
against our lower bound $\sum_{\set{x,y} \in \P} w^{xy}$ on the optimum cost.

To prove that, we consider possible states of \ALG. There are 15 of them. We
assign a~potential $\Psi$ to every state such that for any state
and for any request, it holds that $\Delta \ALG + \Delta \Psi \leq 3 \cdot
\sum_{\set{x,y} \in \P} \Delta w^{xy}$.  Note that $\Psi$ is introduced only to
simplify the proof: it is not related to the potential used internally by 
the algorithm \ALG.

The states are shown in~\autoref{tab:lb3}. In each state we assume that the
list is $abc$; the items are renamed after a transition if
needed. Arrows represent modes: an arrow pointing from an item $x$ towards an
item $y$ indicates that the pair $\{x, y\}$ is in mode $\alpha$ (if $x$ is
before $y$ in the list) or $\gamma$ (otherwise). If there is no arrow between 
$x$ and $y$, the pair is in mode
$\beta$. The column ``move'' specifies what  \ALG  does after the request, while
the column ``dest.'' denotes the index of the destination state after movement. State changes
correspond to the moves of~\ALG. Cost of the pair-based \OPT is denoted $\Delta
\OPT_P = \sum_{\set{x,y} \in \P} \Delta w^{xy}$.

\begin{table}
    
    \begin{tabular}{|c|c|c|c|c|c|c|c|c|}
    \hline
       \textbf{state} & \textbf{idx} & $\Psi$ & \textbf{req.} & \textbf{move} & \textbf{dest.} & $\Delta \ALG$ & $ \Delta \Psi$ & $\Delta \OPT_P$\\
    \hline \hline 
    \multirow{3}{*}{\includegraphics[width=20mm]{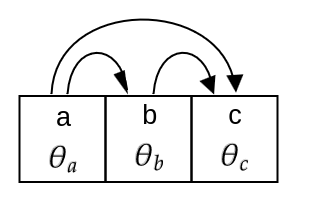}} & \multirow{3}{*}{$0$} & \multirow{3}{*}{$0$} &  $a$& \FRONT & $0$ & $0$ & $0$ & $0$ \\ \cline{4-9} 
                      &                    && $b$ &  $\Theta_b$& $1$ &$1$ &  $\half$ & $\half$ \\ \cline{4-9} 
                      &                    && $c$ &  $\Theta_c$& $2$ & $2$ & $1$ & $1$ \\ \hline
    \multirow{3}{*}{\includegraphics[width=20mm]{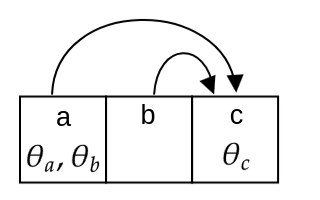}} & \multirow{3}{*}{$1$} & \multirow{3}{*}{$\half$} & $a$ & \FRONT &  $0$& $0$ &$-\half$ & $\half$ \\ \cline{4-9} 
                      &                    && $b$ & \FRONT & $0$ & $2$ &$-\half$ & $\half$ \\ \cline{4-9} 
                      &                    && $c$ & $\Theta_c$ & $3$ &$2$ &$1$ & $1$ \\ \hline
    \multirow{3}{*}{\includegraphics[width=20mm]{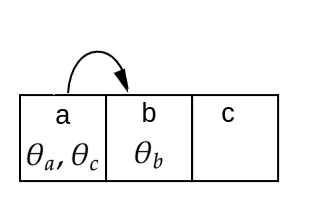}} & \multirow{3}{*}{$2$} & \multirow{3}{*}{$1$} & $a$ & \FRONT & $4$ & $0$ &$- \half$& $\half$ \\ \cline{4-9} 
                      &                    && $b$ & $\Theta_b$ & $5$ & $1$ &$1$ & $1$ \\ \cline{4-9} 
                      &                    && $c$ & \FRONT & $0$ & $4$ & $-1$ & $1$ \\ \hline
    \multirow{3}{*}{\includegraphics[width=20mm]{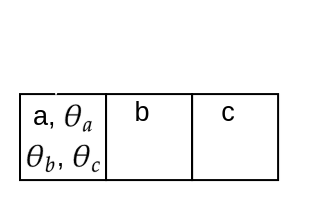}} & \multirow{3}{*}{$3$} &\multirow{3}{*}{$\frac{3}{2}$} & $a$ & \FRONT & $4$ & $0$ & $-1$ & $1$ \\ \cline{4-9} 
                      &                    && $b$ & \FRONT & $4$ & $2$ & $-1$ & $1$ \\ \cline{4-9} 
                      &                    && $c$ & \FRONT & $4$ & $4$ & $-1$ & $1$ \\ \hline
    \multirow{3}{*}{\includegraphics[width=20mm]{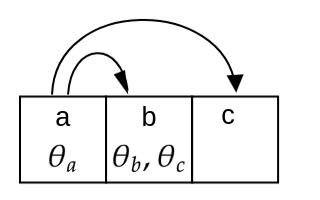}} & \multirow{3}{*}{$4$} & \multirow{3}{*}{$\half$} & $a$ & \FRONT & $4$ & $0$ & $0$ & $0$ \\ \cline{4-9} 
                      &                    && $b$ & $\Theta_b$ & $1$ & $1$ & $0$ & $1$ \\ \cline{4-9} 
                      &                    && $c$ & $\Theta_c$ & $1$ & $3$ & $0$ & $1$ \\ \hline
    \multirow{3}{*}{\includegraphics[width=20mm]{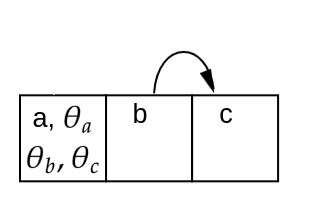}} & \multirow{3}{*}{$5$} & \multirow{3}{*}{$2$} & $a$ & \FRONT & $6$ & $0$ & $0$ & $1$ \\ \cline{4-9} 
                      &                    && $b$ & \FRONT & $4$ & $2$ &$-\frac{3}{2}$ & $\half$ \\ \cline{4-9} 
                      &                    && $c$ & \FRONT & $7$ & $4$ & $-1$ & $1$ \\ \hline
    \multirow{3}{*}{\includegraphics[width=20mm]{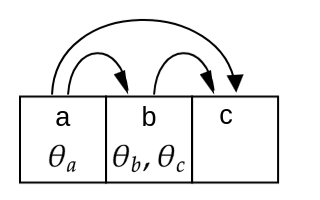}} & \multirow{3}{*}{$6$} &\multirow{3}{*}{$2$} & $a$ & \FRONT &  $6$ & $0$ & $0$ & $0$ \\ \cline{4-9} 
                      &                    && $b$ & $\Theta_b$ &  $1$& $1$&$-\frac{3}{2}$ & $\half$ \\ \cline{4-9} 
                      &                    && $c$ & $\Theta_c$ &  $8$& $3$ &$0$ & $1$ \\ \hline
    \multirow{3}{*}{\includegraphics[width=20mm]{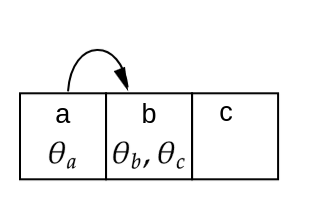}} & \multirow{3}{*}{$7$} & \multirow{3}{*}{$1$} & $a$ & \FRONT & $4$ & $0$&$-\half$ & $\half$ \\ \cline{4-9} 
                      &                    && $b$ & $\Theta_b$ & $9$ & $1$ & $1$& $1$ \\ \cline{4-9} 
                      &                    && $c$ & \FRONT & $0$ & $4$ &$-1$ & $1$ \\ \hline
    \multirow{3}{*}{\includegraphics[width=20mm]{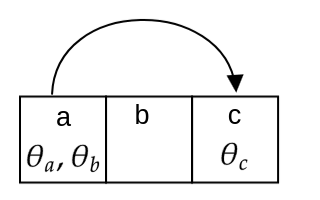}} & \multirow{3}{*}{$8$} & \multirow{3}{*}{$2$} & $a$ & \FRONT & $10$ & $0$ & $-\half$ & $\half$ \\ \cline{4-9} 
                      &                    && $b$ & \FRONT & $0$ &$2$& $-2$ & $1$ \\ \cline{4-9} 
                      &                    && $c$ & $\Theta_c$ & $11$ & $2$& $1$ & $1$ \\ \hline
    \multirow{3}{*}{\includegraphics[width=20mm]{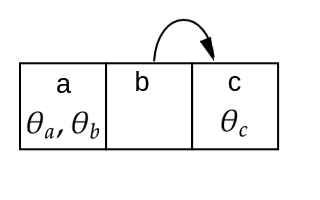}} & \multirow{3}{*}{$9$} & \multirow{3}{*}{$2$} & $a$ & \FRONT & $0$ & $0$ & $ -2$ & $1$ \\ \cline{4-9} 
                      &                    && $b$ & \FRONT & $10$ & $2$ & $-\half$ & $\half$ \\ \cline{4-9} 
                      &                    && $c$ & $\Theta_c$ &  $12$& $2$ & $1$& $1$ \\ \hline
    \multirow{3}{*}{\includegraphics[width=20mm]{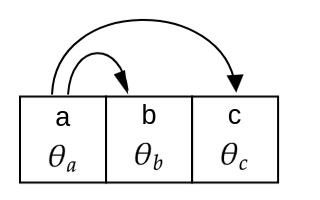}} & \multirow{3}{*}{$10$} & \multirow{3}{*}{$\frac{3}{2}$} & $a$ & \FRONT & $10$ & $0$ & $0$ & $0$ \\ \cline{4-9} 
                      &                     && $b$ & $\Theta_b$ &  $1$& $1$ & $-1$ & $1$ \\ \cline{4-9} 
                      &                     && $c$ & $\Theta_c$ &  $13$& $2$ &  $1$ & $1$ \\ \hline
    \multirow{3}{*}{\includegraphics[width=20mm]{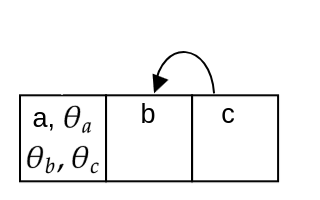}} & \multirow{3}{*}{$11$} & \multirow{3}{*}{$3$} & $a$ & \FRONT &  $14$ &$0$ & $ -1$ & $1$ \\ \cline{4-9} 
                      &                     && $b$ & \FRONT & $7$ &$2$ & $-2$ &$ 1$ \\ \cline{4-9} 
                      &                     && $c$ & \FRONT &$4$  &$4$ & $-\frac{5}{2}$ & $\half$ \\ \hline
    \multirow{3}{*}{\includegraphics[width=20mm]{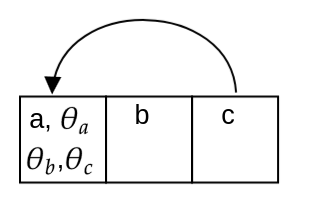}} & \multirow{3}{*}{$12$} & \multirow{3}{*}{$3$} & $a$ & \FRONT & $7$ &$0$ & $ -2$ & $1$ \\ \cline{4-9} 
                      &                     && $b$ & \FRONT & $14$ &$2$ & $-1$ & $1$ \\ \cline{4-9} 
                      &                     && $c$ & \FRONT & $4$ & $4$ & $-\frac{5}{2}$ & $\half$ \\ \hline
    \multirow{3}{*}{\includegraphics[width=20mm]{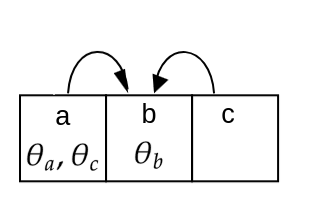}} & \multirow{3}{*}{$13$} & \multirow{3}{*}{$\frac{5}{2}$} & $a$ & \FRONT & $14$ & $0$ &$-\half$ & $\half$ \\ \cline{4-9} 
                      &                     && $b$ & $\Theta_b$ & $3$ &$1$ & $-1$ & $1$ \\ \cline{4-9} 
                      &                     && $c$ & \FRONT & $0$ &$4$ & $-\frac{5}{2}$ & $\half$ \\ \hline
    \multirow{3}{*}{\includegraphics[width=20mm]{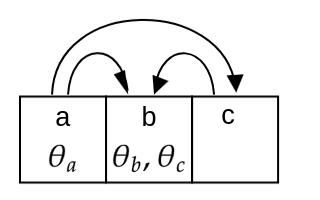}} & \multirow{3}{*}{$14$} &\multirow{3}{*}{$2$} & $a$ & \FRONT & $14$ & $0$ &$0$& $0$ \\ \cline{4-9} 
                      &                     && $b$ & $\Theta_b$ &  $8$ & $1$ & $0$ & $1$ \\ \cline{4-9} 
                      &                     && $c$ & $\Theta_c$ & $1$  &$3$ & $-\frac{3}{2}$ & $\half$ \\ \hline
        
    \end{tabular}
    \caption{\ALG is 3-competitive for lists of length 3}
    \label{tab:lb3}
\end{table}

\medskip

There exists an easier $3$-competitive algorithm for $3$ items. This algorithm
uses $6$ states to keep track of the past history. The algorithm is illustrated
below as a graph whose vertices represent the current state of the algorithm and
edges represent its transitions. Again, items are renamed after each transition
so that the list is always $abc$. The modes of each pair are
represented the same way as in~\autoref{tab:lb3}. Each edge is labeled by the
requested item, the algorithm's cost, and the optimum cost ($\half$ for each
mode change). The algorithm's move is implied by the cost value. For example, if
the request is $c$ and the cost is $3$, it means that the algorithm
paid $2$ to access  $c$ and then swapped $c$ with the
preceding $b$.
 
\begin{center}
\includegraphics[width=4.5in]{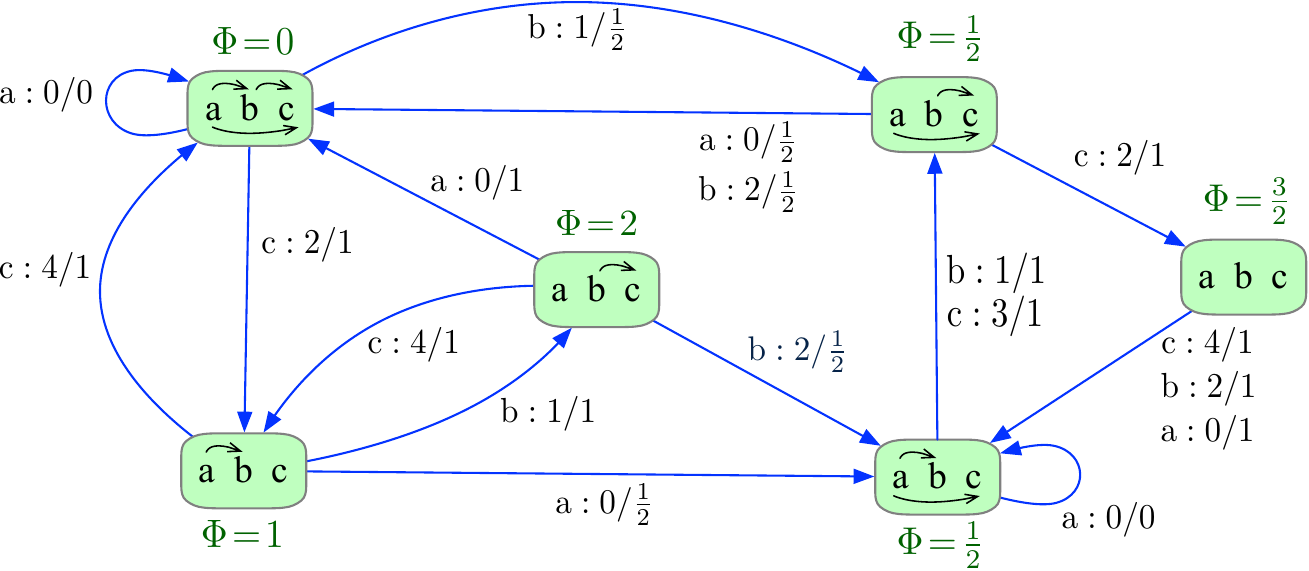}
\end{center}

The diagram also shows the potential values $\Phi$ for each state. By
routine verification, for each transition we have that the amortized
cost (that is, the cost of the algorithm plus the potential change) is at most $3$ times the optimum cost.


\section{Lower bound for the \textsc{Full-Or-Partial-Move} algorithm}
\label{sec:lower-fpm}

In this section, we show that for the choice of parameters
from~\autoref{lem:amortized_cost_P1}, the competitive ratio of \ALG is at least
$3.04$ in the partial cost model.

The sequence achieving this ratio is shown
in~\autoref{tab:pm_lower}. It uses five items (on a longer list one may simply 
use only five initial items). Graphs in the column
``work function on pairs'' depict modes of all the pairs. There, an arrow from $x$
to $y$ means that $W^{xy}(xy) + 1 = W^{xy}(yx)$  (which corresponds to 
mode $\alpha$ or $\gamma$,
depending on the order of the list). Some arrows are omitted for clarity:
the actual graph is the transitive closure of the graph shown.

Note that the final state of the list after the sequence is the same as the
state before the sequence, therefore it can be repeated indefinitely. The first
state described in the table is accessible from the starting state described
in~\autoref{sec:algorithm} by requesting item $d$ and then item~$a$.

\autoref{tab:pm_lower} compares cost of \ALG to the cost of \emph{pair-based} \OPT
defined as $P_{OPT}= \sum_{\set{x,y} \in \P} \Delta w^{xy}$, which for this
sequence is equal to $25$. We now show the behavior of an actual offline algorithm \OPT
that achieves this cost. Before the start of the sequence, the algorithm changes
its list to $cdeab$. This incurs an additional additive cost, which becomes
negligible if the sequence is repeated sufficiently many times. During the
entire sequence, \OPT performs two moves: after serving the first
request, it moves $c$ to the third position in the list, and after the ninth
request, it brings $c$ back to the front. Thus, it executes $4$ swaps.
It is easy to check that the cost of requests for paid by \OPT is $21$, and hence its 
total cost is~$25$.

\begin{table}
   \begin{tabular}[t]{|c|c|c|c|c|c|}
   \hline
      \multirow{2}{*}{\textbf{list of} \ALG} & \textbf{work function} & \multirow{2}{*}{\textbf{req.}} 
         & \textbf{move} & \textbf{cost} & \textbf{cost} \\ 
      & \textbf{on pairs} & & \textbf{of} \ALG & of \ALG & \textbf{of} $\OPT_P$ \\
      \hline\hline
      
      \multirow{2}{*}{\includegraphics[scale=0.8]{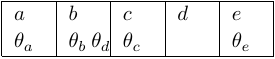}}
      & \multirow{2}{*}{\includegraphics[scale=0.22]{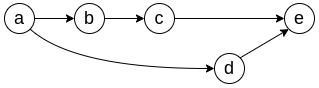}} 
      & & & & \\
      \cline{3-6}

      & & \multirow{2}{*}{$c$} & \multirow{2}{*}{$\theta_c$} & \multirow{2}{*}{2} & \multirow{2}{*}{$1\half$} \\
      \cline{1-2}
      
      \multirow{2}{*}{\includegraphics[scale=0.8]{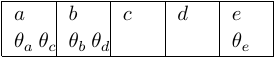}}
      & \multirow{2}{*}{\includegraphics[scale=0.22]{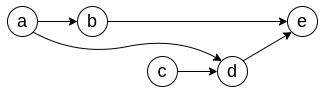}} 
      & & & & \\
      \cline{3-6}

      & & \multirow{2}{*}{$e$} & \multirow{2}{*}{$\theta_e$} & \multirow{2}{*}{6} & \multirow{2}{*}{$3\half$} \\
      \cline{1-2}
      
      \multirow{2}{*}{\includegraphics[scale=0.8]{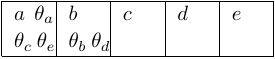}}
      & \multirow{2}{*}{\includegraphics[scale=0.22]{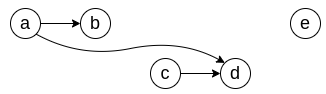}} 
      & & & & \\
      \cline{3-6}

      & & \multirow{2}{*}{$e$} & \multirow{2}{*}{\FRONT} & \multirow{2}{*}{14} & \multirow{2}{*}{$5\half$} \\
      \cline{1-2}
      
      \multirow{2}{*}{\includegraphics[scale=0.8]{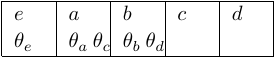}}
      & \multirow{2}{*}{\includegraphics[scale=0.22]{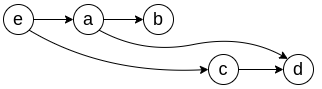}} 
      & & & & \\
      \cline{3-6}

      & & \multirow{2}{*}{$d$} & \multirow{2}{*}{$\theta_d$} & \multirow{2}{*}{20} & \multirow{2}{*}{$7\half$} \\
      \cline{1-2}
      
      \multirow{2}{*}{\includegraphics[scale=0.8]{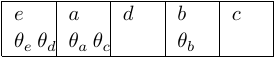}}
      & \multirow{2}{*}{\includegraphics[scale=0.22]{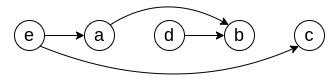}} 
      & & & & \\
      \cline{3-6}

      & & \multirow{2}{*}{$c$} & \multirow{2}{*}{\FRONT} & \multirow{2}{*}{28} & \multirow{2}{*}{$9\half$} \\
      \cline{1-2}
      
      \multirow{2}{*}{\includegraphics[scale=0.8]{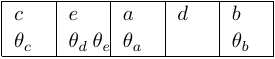}}
      & \multirow{2}{*}{\includegraphics[scale=0.2]{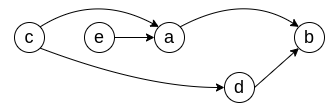}} 
      & & & & \\
      \cline{3-6}

      & & \multirow{2}{*}{$d$} & \multirow{2}{*}{$\theta_d$} & \multirow{2}{*}{33} & \multirow{2}{*}{11} \\
      \cline{1-2}
      
      \multirow{2}{*}{\includegraphics[scale=0.8]{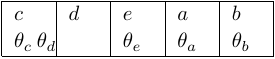}}
      & \multirow{2}{*}{\includegraphics[scale=0.22]{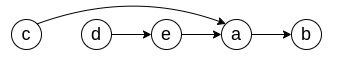}} 
      & & & & \\
      \cline{3-6}

      & & \multirow{2}{*}{$e$} & \multirow{2}{*}{$\theta_e$} & \multirow{2}{*}{35} & \multirow{2}{*}{12} \\
      \cline{1-2}
      
      \multirow{2}{*}{\includegraphics[scale=0.8]{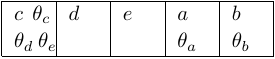}}
      & \multirow{2}{*}{\includegraphics[scale=0.22]{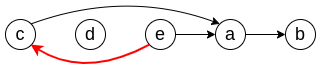}} 
      & & & & \\
      \cline{3-6}
      
      & & \multirow{2}{*}{$e$} & \multirow{2}{*}{\FRONT} & \multirow{2}{*}{39} & \multirow{2}{*}{$12\half$} \\
      \cline{1-2}
      
      \multirow{2}{*}{\includegraphics[scale=0.8]{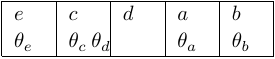}}
      & \multirow{2}{*}{\includegraphics[scale=0.22]{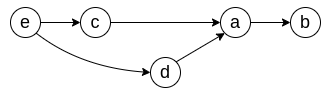}} 
      & & & & \\
      \cline{3-6}

      & & \multirow{2}{*}{$d$} & \multirow{2}{*}{$\theta_d$} & \multirow{2}{*}{42} & \multirow{2}{*}{$13\half$} \\ 
      \cline{1-2}
      
      \multirow{2}{*}{\includegraphics[scale=0.8]{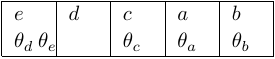}}
      & \multirow{2}{*}{\includegraphics[scale=0.22]{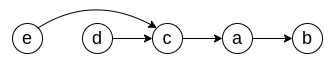}} 
      & & & & \\
      \cline{3-6}

      & & \multirow{2}{*}{$c$} & \multirow{2}{*}{$\theta_c$} & \multirow{2}{*}{44} & \multirow{2}{*}{$14\half$} \\ 
      \cline{1-2}
      
      \multirow{2}{*}{\includegraphics[scale=0.8]{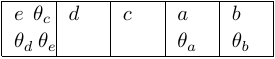}}
      & \multirow{2}{*}{\includegraphics[scale=0.22]{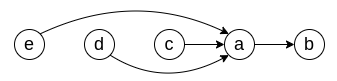}} 
      & & & & \\
      \cline{3-6}

      & & \multirow{2}{*}{$c$} & \multirow{2}{*}{\FRONT} & \multirow{2}{*}{48} & \multirow{2}{*}{$15\half$} \\ 
      \cline{1-2}
      
      \multirow{2}{*}{\includegraphics[scale=0.8]{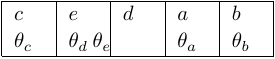}}
      & \multirow{2}{*}{\includegraphics[scale=0.22]{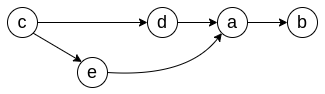}} 
      & & & & \\
      \cline{3-6}

      & & \multirow{2}{*}{$b$} & \multirow{2}{*}{$\theta_b$} & \multirow{2}{*}{52} & \multirow{2}{*}{$17\half$} \\ 
      \cline{1-2}
      
      \multirow{2}{*}{\includegraphics[scale=0.8]{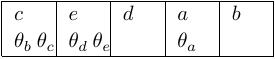}}
      & \multirow{2}{*}{\includegraphics[scale=0.22]{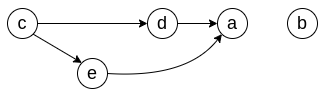}} 
      & & & & \\
      \cline{3-6}

      & & \multirow{2}{*}{$b$} & \multirow{2}{*}{\FRONT} & \multirow{2}{*}{60} & \multirow{2}{*}{$19\half$} \\ 
      \cline{1-2}
      
      \multirow{2}{*}{\includegraphics[scale=0.8]{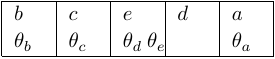}}
      & \multirow{2}{*}{\includegraphics[scale=0.22]{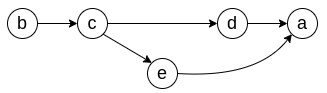}} 
      & & & & \\
      \cline{3-6}

      & & \multirow{2}{*}{$d$} & \multirow{2}{*}{$\theta_d$} & \multirow{2}{*}{64} & \multirow{2}{*}{21} \\ 
      \cline{1-2}
      
      \multirow{2}{*}{\includegraphics[scale=0.8]{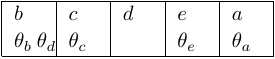}}
      & \multirow{2}{*}{\includegraphics[scale=0.22]{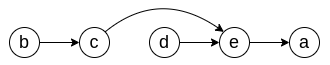}} 
      & & & & \\
      \cline{3-6}

      & & \multirow{2}{*}{$a$} & \multirow{2}{*}{$\theta_a$} & \multirow{2}{*}{68} & \multirow{2}{*}{23} \\ 
      \cline{1-2}
      
      \multirow{2}{*}{\includegraphics[scale=0.8]{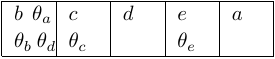}}
      & \multirow{2}{*}{\includegraphics[scale=0.22]{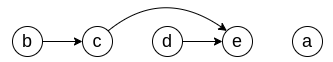}} 
      & & & & \\
      \cline{3-6}

      & & \multirow{2}{*}{$a$} & \multirow{2}{*}{\FRONT} & \multirow{2}{*}{76} & \multirow{2}{*}{25} \\ 
      \cline{1-2}
      
      \multirow{2}{*}{\includegraphics[scale=0.8]{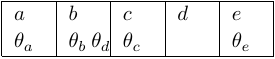}}
      & \multirow{2}{*}{\includegraphics[scale=0.22]{FPM_lb_pictures/graf0.png}} 
      & & & & \\
      \cline{3-6}
      & & & & & \\
      \hline
   \end{tabular}
    \caption{Sequence achieving ratio $3.04$ for \ALG. The red arrow indicates the pair in mode $\gamma$.}
    \label{tab:pm_lower}
  \end{table}

\section{Computational Results}
\label{sec:computational}

This section contains the description of our computational study of the
$\probLUPuniform$ problem. 
Our source code for reproducing the results is available at \cite{2025_github}. 

\subparagraph*{Note on the accuracy of our results.} Our computational results
serve as sources of insight into the $\probLUPuniform$ problem, and our code
produces graphs which can be checked for correctness, although the size of the
graph may be an obstacle. 

We therefore list our results as \emph{computational results}, not as theorems,
and advise the reader to be similarly cautious when interpreting them or citing
them in future work.

\subsection{Upper bound of 3 for list length up to 6}
\label{subsec:computational-upper-bound}

We start by discussing our results of a exhaustive computer search for the optimal competitive ratio for the
$\probLUPuniform$ model for small list lengths.

\begin{compresult}
There exists a $3$-competitive algorithm for $\probLUPuniform$
for list size $n \leq 6$.
\end{compresult}

\smallskip
Our computational method starts with a fixed competitive ratio $3$ and creates a bipartite graph.
In one part (we call it the \OPT part), we store all pairs $(\sigma, W)$,
where $\sigma$ is the current list of the algorithm and $W$ is a work function for this state.
(Recall that the work function~$W$ is the optimal cost of serving the input $\sigma$ and
ending in a given list state.)
Each vertex in this part has $n$ edges, corresponding to the new request $r$. Each edge leads
to a~vertex in the \AALG part of the graph, which stores a triple ($\sigma$, $W'$, $r$),
where $W'$ is the new work function after including the new request $r$ into the sequence.
In turn, each \AALG vertex has $n!$ edges, each leading to a different permutation $\sigma'$,
which we can interpret as the new list of the algorithm after serving the request $r$.

Edges of the graph have an associated cost with them, where the cost of edges leaving \AALG correspond
to the cost of the algorithm serving the request $r$ and the cost of edges leaving \OPT correspond
to the increase of the minimum of the work function.

Each vertex in this graph stores a number, which we call the \emph{potential} of the \AALG or \OPT vertex.
We initialize the potentials to all-zero values and then run an iterative procedure, inspired
by \cite{1991_chrobak_larmore_server_problem}, which iteratively computes a new minimum
potential in each \OPT vertex and a new maximum potential in each \AALG vertex.

If this procedure terminates, we have a valid potential value for each vertex in the graph,
and by restricting our algorithm to pick any tight edge leaving an \AALG vertex, we obtain an~algorithm
that is $3$-competitive. The data generated after this termination are similar in nature
to \autoref{tab:lb3}.

Our computational results correspond to our procedure terminating for list sizes $4 \leq n \leq 6$.

For list size $n=6$, it is no longer feasible to create a graph of all reachable work functions,
as our computations indicate that even listing all of these would require more memory
than our current computational resources allow. We therefore switch to the model of pair-based
\OPT of \autoref{sec:preliminaries}. As explained in that section, a $c$-competitive algorithm
for pair-based \OPT is also $c$-competitive against regular \OPT. However, the space of reachable
pair-based work functions is much smaller, so the setting of $n=6$ can be explored computationally
as well.

Note that we can also switch the results of $n = 4$ and $n = 5$ into the pair based setting,
which improves the running time. The computational results against the pair-based
OPT in those cases are consistent with our results against general \OPT.

\subsection{Lower bound for the work function algorithm class}
\label{subsec:computational-lower-bound}

We can translate the algorithms for \probLUPstandard from
\cite{1999_anderson_lup_work_function} to the \probLUPuniform model. For input
$\sigma$ to \probLUPuniform and a list $\pi$, the work function
$W^{\sigma}(\pi)$ defines the optimal cost of serving the input $\sigma$ and
ending with the list $\pi$. Then, a work function algorithm, given its current
state of the list $\mu$ and a new request $r$, it first serves the request and
then reorders the list to end in a state $\pi$ being a minimizer of 
$W^{\sigma,r}(\pi) + d(\pi, \mu)$, where $d(\pi, \mu)$ is the number of
pairwise swaps needed to switch from $\mu$ to $\pi$. The minimizer state is not
necessarily unique, hence we talk about the \emph{work function algorithm
class}.

Recall that some algorithms from
the work function algorithm class are $2$-competitive for \probLUPstandard~\cite{1999_anderson_lup_work_function},
which means they are optimal for that setting. In contrast, we show the following.

\begin{compresult}
The competitive ratio of all algorithms
in the work function algorithm class for \probLUPuniform is at least $3.1$ (for the partial cost model).
\end{compresult}

\smallskip

To show this claim, we construct a lower bound instance graph for list length $5$ where all
choices of the work function algorithm class are being considered. This graph is equivalent
to the graph construction of \autoref{subsec:computational-upper-bound}, with the out-edges
of the \AALG vertices restricted only to all valid choices of the work function algorithm class.

However, since the result of an iterative algorithm is that the potential does
not stabilize, we constructed a tailored approach for producing a lower bound
graph from the unstable potentials. A result of this code is then a standard
lower bound using an adversary/algorithm graph, with the adversary presenting a
single item in each of its vertices and the competitive ratio of any cycle in
the graph being at least $3.1$. 

We remark that this approach yields a $3$-competitive work function algorithm for 
\probLUPuniform for list lengths $3$ and $4$.


\end{document}